\newif\ifpdf
\documentclass[12pt,onecolumn,draftcls]{IEEEtran}
\usepackage{setspace}
\ifpdf
 \usepackage[pdftex]{graphicx}
\else
 \usepackage{graphicx}
\fi
\usepackage{url}
\usepackage[cmex10]{amsmath}
\usepackage{amsmath}
\usepackage{amssymb}
\usepackage{latexsym}
\usepackage{epsfig}
\usepackage{graphicx}
\usepackage{color}
\usepackage{setspace}
\usepackage{stmaryrd}

\usepackage{array}
\newcolumntype{C}[1]{>{\centering\let\newline\\\arraybackslash\hspace{0pt}}m{#1}}
\newcolumntype{L}[1]{>{\raggedright\let\newline\\\arraybackslash\hspace{0pt}}m{#1}}
\newcolumntype{R}[1]{>{\raggedleft\let\newline\\\arraybackslash\hspace{0pt}}m{#1}}
\usepackage{multirow}

\usepackage{enumitem}

\newcommand\mb[1]{\mathbf{#1}}

\newcommand{\defeq}{\stackrel{def}{=}}

\newtheorem{theorem}{Theorem}
\newtheorem{lemma}{Lemma}
\newtheorem{corollary}{Corollary}

\newtheorem{definition}{Definition}
\newtheorem{example}{Example}

\newtheorem{remark}{Remark}
\newcommand{\suchthat}{%
      \mathrel{\ooalign{$\ni$\cr\kern-1pt$-$\kern-6.5pt$-$}}}
\title{New Constructions of Zero-Correlation Zone Sequences}
\author{Yen-Cheng Liu, Ching-Wei Chen, and Yu T. Su
\thanks{Y.-C. Liu and Y. T. Su (correspondence addressee) are with
the Institute of Communications Engineering, National Chiao Tung
University, Hsinchu, Taiwan (email: ycliu@ieee.org;
ytsu@nctu.edu.tw). C.-W. Chen is with National Instruments Taiwan
Corp., Taipei, Taiwan (email: penguinjazzy@gmail.com). The
material in this paper was presented in part at the IEEE 2009
International Symposium on Information Theory.}
        }

\begin{document}
\ifpdf
 \DeclareGraphicsExtensions{.pdf,.jpg}
\else
 \DeclareGraphicsExtensions{.eps}
\fi \maketitle

\begin{abstract}
In this paper, we propose three classes of systematic approaches
for constructing zero correlation zone (ZCZ) sequence families. In
most cases, these approaches are capable of generating sequence
families that achieve the upper bounds on the family size ($K$)
and the ZCZ width ($T$) for a given sequence period ($N$).

Our approaches can produce various binary and polyphase ZCZ
families with desired parameters $(N,K,T)$ and alphabet size. They
also provide additional tradeoffs amongst the above four system
parameters and are less constrained by the alphabet size. Furthermore,
the constructed families have nested-like property that can be either
decomposed or combined to constitute smaller or larger ZCZ
sequence sets. We make detailed comparisons with related works and
present some extended properties. For each approach, we provide
examples to numerically illustrate the proposed construction
procedure.
\end{abstract}


\IEEEpeerreviewmaketitle
\newpage
%\setstretch{0.95}
\section{Introduction}
\IEEEPARstart{F}{amilies of} sequences with some desired periodic
or aperiodic autocorrelation (AC) and cross-correlation (CC)
properties are useful in communication and radar systems for
applications in identification, synchronization, ranging, or/and
interference mitigation. For example, to minimize the multiple
access interference (MAI) and self-interference (e.g.,
inter-symbol interference) in a multi-user, multi-path environment
or to avoid inter-antenna interference in a multiple-input,
multiple-output system, one would like to have an {\it ideal
sequence set} whose periodic AC functions are nonzero only at the
zeroth correlation lag ($\tau=0$) and whose pairwise periodic CC
values are identically zero at any $\tau$ for all pairs of
sequences. Similar aperiodic properties are called for in
designing pulse compressed radar signal or two-dimensional array
waveforms to have an impulse-like ambiguity function satisfying
the resolution requirements.

Unfortunately, the ideal sequence set does not exist, i.e., it is
impossible to have impulse-like AC functions and zero CC functions
simultaneously in a sequence set. In fact, bounds on the magnitude
of CC and AC values derived in \cite{Welch} and \cite{Sarwate}
suggest that the design of sequence sets involves the tradeoff
between AC and CC values. An alternate compromise is to require
that the ideal AC and CC properties be maintained only at
correlation lags within a window called zero-correlation zone
(ZCZ) \cite{1st_ZCZ}. Sequences with such properties are known as
ZCZ sequences. Little or no system performance degradation results
if the correlation values outsides the ZCZ are immaterial to the
application of concern. For example, if the maximum channel delay
spread $T_m$ and the maximum distance between a base station and
co-channel users $D_m$ are known, a direct sequence spread
spectrum based multiple access system using a family of ZCZ
sequences with ZCZ width $|\tau|\leq T_m+2D_m/c$, where $c$ is the
speed of light, will be able to suppress MAI and multipath self
interference.

Other than the restrictions on the magnitude of correlation
values, practical implementation concerns prefer that the choice
of the sequence period be flexible and the family size be as large
as possible while keeping the desired AC and CC properties intact.
One also hope that the elements of the sequences be drawn from an
alphabet set as small as possible.

Various ZCZ sequence generation methods have been proposed
\cite{comp_ex_1}--\cite{8QAM}. The methods presented in
\cite{comp_ex_1}--\cite{comp_ex_4} are based on complementary
sequence sets. Interleaving techniques are shown to be effective
in constructing ZCZ sequences \cite{Polyphase}--\cite{interleave_Gong}.
They can be generalized to construct two-dimensional ($2$-D) ZCZ
arrays \cite{Array1,Array2} as well. Sets of ZCZ sequences derived
from manipulating perfect sequences were suggested in
\cite{perfect_seq_based} and \cite{GCL}. Park \textit{et al.}
\cite{PS} construct sequences that has nonzero AC only at
subperiodic correlation lags and zero CC across all lags.
By requiring the transform domain sequences to satisfy some
special properties, \cite{Transform,Zak_trans} present methods
that generate ZCZ sequences having zero CC across all lags.
Some ZCZ sequence sets can be partitioned into smaller subsets so
that the zero-CC zone of any two sequences drawn from different subsets are wider than that
among intra-subset sequences. Ternary or polyphase sequences
with such a property have been constructed via interleaving
techniques \cite{AZCZ2,AZCZ3} and in \cite{16QAM,8QAM} quadrature
amplitude-modulated (QAM) sequences are shown to be derivable from
binary or ternary sequences.

In this paper, we present three systematic approaches for
generating families of sequences whose periodic AC and CC
functions satisfy a variety of ZCZ requirements. While some known
ZCZ sequence construction methods employing Hadamard matrices in
time domain (e.g., \cite{interleave_Gong,GCL}), our first approach
uses such matrices to meet the desired transform domain properties
of a ZCZ sequence set instead. Sequence sets generated from this
approach are, by construction, optimal in the sense that the upper
bounds for family sizes and ZCZ widths are achieved. We further
employ a filtering operation to convert sequences of nonconstant
modulus symbols into polyphase ones without changing the
correlation properties.

Based upon a basic binary sequence (to be defined in Section
\ref{section:direct}) whose AC function satisfies the ZCZ
requirement, the second approach generates ZCZ sequence families
by a special nonuniform upsampling on unitary matrices. The
construction of basic sequences seems trivial and straightforward,
but from these simple sequences we are able to synthesize desired
polyphase ZCZ sequences through some refining steps that include
nonuniform upsampling and filtering.

Our third approach invokes the notion of complementary set of
sequences \cite{generalcom,complementary}. It bears the flavor of
the second approach and makes use of a basic binary sequence which
meets the ZCZ constraint as well as a collection of mutually
orthogonal complementary sets. While this method is capable of
generating binary sequences with sequence parameters identical to
those given in \cite{comp_ex_1} and \cite{comp_ex_4}, it can also
produce polyphase sequence sets that are unobtainable by the
conventional complementary set-based approaches.

The rest of this paper is organized as follows. We introduce basic
definitions and properties related to our investigation in the
next section. Section \ref{section:transform} begins with a brief
summary of important transform domain properties, followed by the
analysis and synthesis of the proposed transform domain approach.
We then show some ZCZ sequence sets generated by the transform
domain method in subsection \ref{section:trans_ex}. The direct
synthesis method is presented in section \ref{section:direct} and
construction examples are given in subsection
\ref{section:examples}. In section \ref{section:comp}, a
complementary sequence set based extension of the second approach
is proposed, followed by numerical construction examples given in
subsection \ref{section:CSS_examples}. For each proposed approach,
we tabulate the parametric constraint comparisons with related
methods. More detailed comparisons and discussions are given in
the form of remarks. Finally, some concluding notes are provided
in Section \ref{section:conclusion}.

%\setstretch{1}

\section{Definitions and Fundamental Properties}
%\begin{IEEEeqnarray}{rCl}
% a&=&a+b+c\\
% &&+\:d+e
%\end{IEEEeqnarray}
%
\begin{definition}
An $(N,K)$ sequence set $\mathbf{X}$ is a set of $K$ sequences of
period $N$.
\end{definition}

\begin{definition}
The periodic CC function of two period-$N$ sequences $u \equiv
\{u(n)\}$ and $v \equiv \{u(n)\}$ is defined as
\begin{equation}
\theta_{uv}(\tau)=\sum_{n=0}^{N-1}u(n)v^*(n-\tau)=u(\tau)\varoast%\otimes
        v^{*}(-\tau),
\end{equation}
where $\varoast$ denotes the circular convolution.
\end{definition}
Thus, the periodic AC function of sequence $u$ is simply
$\theta_{uu}(\tau)$. Since these CC and AC functions are also of
period $N$, to simplify the discussion we shall, throughout this
paper, limit the representations and examples of sequences or
sequence sets to a single period $(0\leq\tau\leq N-1)$ unless
necessary.

%\begin{definition}
%An $(N,K)$ sequence set $\mathbf{X}$ is said to be \textit{ideal} %(sequence) set
%if and only if the following two conditions are satisfied
%simultaneously:
%\begin{enumerate}
%    \item[i.]{$\theta_{uv}(\tau)=0$, $\forall \tau$ and $\forall u,v\in \mathbf{X}$, $u\neq v$;}
%    \item[ii.]{$\theta_{uu}(\tau)= \theta_{uu}(0)\delta(\tau)$, $\forall u
%\in \mathbf{X}$, where
%\begin{eqnarray}
%\delta(\tau)\defeq\left\{
%\begin{array}{ll}
%1 , & \tau=0; \\
%0 , & \tau \neq 0. \\
%\end{array}
%\right.
%\end{eqnarray}}
%\end{enumerate}
%\end{definition}
%As mentioned previously, it is impossible to have an ideal
%sequence set, or equivalently, to have both perfect AC and CC
%properties for all correlation lags simultaneously in a set.% This will be seen later in this section.

\begin{definition}
%A sequence that poses such AC function
A sequence $\{u(n)\}$ that has an impulse-like (or ideal) AC
function, i.e., $\theta_{uu}(\tau)=\theta_{uu}(0)\delta(\tau)$, is
called a \textit{perfect sequence}, where
\begin{eqnarray}
\delta(\tau)\defeq\left\{
\begin{array}{ll}
1 , & \tau=0; \\
0 , & \tau \neq 0. \\
\end{array}
\right.
\end{eqnarray}
is the Dirac delta function.
\end{definition}

\begin{definition}
A sequence $\{u_v(n)\}$ is said to be obtained from
\textit{filtering} the sequence $u=\{u(n)\}$ by the sequence
$v=\{v(n)\}$ of the same period if
\begin{eqnarray}
{u}_v(n)\defeq u(n)\circ v(n)\defeq u(n) \varoast v^*(-n)\equiv
\theta_{uv}(n)
\end{eqnarray}
\label{def:mod}
\end{definition}
\begin{definition}
An $(N,K)$ sequence set, $\mathbf{C}=\{C_{0},C_{1},\cdots,
C_{K-1}\}$ is called an $(N,K,T)$ ZCZ sequence family (or set) if
$\forall~ C_i, C_j \in {\mathbf C}$, $i \neq j$,
$\theta_{C_iC_j}(\tau) = 0$ and
$\theta_{C_iC_i}(\tau)=\theta_{C_iC_i}(0)\delta(\tau)$,
$|\tau|_N\leq T<N$ where $T$ is the ZCZ width and
$|k|_N\stackrel{def}{=}k$ mod $N$. \label{def:ZCZ}
\end{definition}
In \cite{ZCZbound}, it was proved that
\begin{lemma}
The sequence period $N$, cardinality $K$ and ZCZ width $T$ of an
$(N,K,T)$ ZCZ family must satisfy the inequality
\begin{eqnarray}
K (T+1) \leq N. \label{eq:bound}
\end{eqnarray}
For $\pm 1$-valued binary sequence set, the bound becomes more
tight \cite{Polyphase}
\begin{eqnarray}
K T \leq \frac{N}{2},~K > 1. \label{eq:bi_bound}
\end{eqnarray}
\end{lemma}
This lemma describes the fundamental tradeoff among the sequence
period, family size, and ZCZ width. For a fixed $N$, increasing
the family size must be achieved at the cost of reduced ZCZ width
and vice versa. Note that for a set with a single perfect
sequence, (\ref{eq:bound}) is automatically satisfied because
$K=1$ and $T=N-1$.

\begin{definition}
An $N\times N$ matrix $\mathbf{U}$ is called a \textit{Hadamard
matrix} of order $N$ if and only if it satisfies two conditions:
\begin{enumerate}
\item[(i)]{Unimodularity: the components of $\mathbf{U}$ are of the
same magnitude $\sqrt{P}$;} % of magnitude $1$ (if proof is needed)
\item[(ii)]{Orthogonality:
$\mathbf{U}\mathbf{U}^H=NP\mathbf{I}_N$ %=\mathbf{U}^H\mathbf{U}
where $\mathbf{I}_N$ is the $N\times N$ identity matrix and
$(\cdot)^H$ denotes the conjugate transpose of the enclosed
matrix.}
\end{enumerate}
\end{definition}

\begin{definition}
The Matrix
\begin{eqnarray}
\mathbf{F}_M=\left[
\begin{array}{ccccc}
1 & 1 & 1 & \cdots & 1 \\
1 & W_{M}^{-1} & W_{M}^{-2} & \cdots & W_{M}^{-(M-1)} \\
1 & W_{M}^{-2} & W_{M}^{-4} & \cdots & W_{M}^{-2(M-1)} \\
\vdots    &  \vdots   &  \vdots   &   \ddots &        \vdots         \\
1 & W_{M}^{-(M-1)} & W_{M}^{-2(M-1)} & \cdots &W_{M}^{-(M-1)^2} \\
\end{array}
\right]^{}%\nonumber
\label{eq:DFT_matrix}
\end{eqnarray}
is called the $M$-discrete Fourier transform ($M$-DFT) matrix,
where $W_M^k=e^{j2\pi k/M}$, and its Hermitian $\mathbf{F}_M^{H}=
\mathbf{F}_M^{-1}$ is called the inverse $M$-DFT ($M$-IDFT)
matrix. The set of complex $M$th roots of unity, $\{W_M^k:
k=0,1,\cdots,M-1\}$, is called the $M$-ary phase-shift keying
($M$-PSK) set and a sequence with elements from the $M$-PSK
constellation is called an $M$-PSK sequence or a polyphase
sequence in general.
\end{definition}
Note that DFT matrices form a subcalss of the so-called Butson
Hadamard matrices \cite{Hada_ref}.
\begin{definition}
The $k$th \textit{Kronecker power} of matrix $\mathbf{U}$, denoted
by $\otimes^k \mathbf{U}$, is defined as
\begin{eqnarray}
\otimes^k\mathbf{U}=\underbrace{\mathbf{U}\otimes\mathbf{U}\otimes\cdots\otimes\mathbf{U}}_{\text{%matrix
      $\mathbf{U}$ appears $k$ times}},
\label{def:kron_prod}
\end{eqnarray}
where $\otimes$ denotes the Kronecker product.
\end{definition}
\begin{definition}
The matrices
\begin{eqnarray} \mathbf{H}_2=\left[
\begin{array}{rr}
1 & 1 \\
1 & -1 \\
\end{array}
\right]
\end{eqnarray}
and
\begin{eqnarray} \mathbf{H}_{2^n}=\otimes^n \mathbf{H}_2=\left[
\begin{array}{rr}
\mathbf{H}_{2^{n-1}} & \mathbf{H}_{2^{n-1}} \\
\mathbf{H}_{2^{n-1}} & -\mathbf{H}_{2^{n-1}} \\
\end{array}
\right],~ n=2,3,\cdots, \label{eq:recur_gen}
\end{eqnarray}
are called \textit{Sylvester Hadamard matrices}.
\label{def:std_Hadamard_mtx}
\end{definition}

%It is easy to show that
The following lemma is essential to derive our construction
methods in the next section.
\begin{lemma}\cite{Hada_ref}
The Kronecker (tensor) product of any two Hadamard matrices is a
Hadamard matrix. \label{Hada_Hada_Hada}
\end{lemma}

\section{Transform Domain Construction Methods}
\label{section:transform} We first review some transform domain
properties of sequences and their correlation functions. A class
of ZCZ sequence construction approaches based on transform domain
properties is then presented. Detailed comparisons with two
related proposals are made and a few construction examples are
provided.
\subsection{Useful Transform Domain Properties}%%%%%%%%%%%%Preliminaries
Denote by DFT$\{u(n)\}$, the DFT of a periodic sequence $\{u(n)\}$
and by IDFT$\{U(k)\}$, the inverse DFT (IDFT) of a periodic
transform domain sequence $\{U(k)\}$. We then immediately have
%The following fundamental lemmas and corollaries are needed in the
%subsequent discourse.
\begin{lemma}
%%%%%%%%%%%%%%%%%%%%%%%%%%%%%%%%%%%%%%%%%%%%%%%%%%%%%%%%%%%%%%%%%%%%%%%%%%%%%%%%參考投影片
The DFT of the CC function $\theta_{uv}(\tau)$ of two period-$N$
sequences, $\{u(n)\}$ and $\{v(n)\}$, is equal to $U(k)V^{*}(k)$,
where $\{U(k)\}=$ DFT$\{u(n)\}$ and $\{V(k)\}=$ DFT$\{v(n)\}$.
\label{CC_DFT}
\end{lemma}

Since the AC function of $\{u(n)\}$ can be expressed as
$\theta_{uu}(n)=u(n)\varoast u^{*}(-n)$, its DFT is given by
$\Theta_{uu}(k)=|U(k)|^{2}$. Therefore, it is straightforward to
show
\begin{corollary}
Sequence $\{u(n)\}$ is a perfect sequence if and only if
$|U(k)|^{2}$ is constant for all $k$. \label{perfect_seq}
\end{corollary}
Based on the above properties, we can easily prove that
%We apply the above results to derive an important property of the
%modulation operator. Consider two sequences $U$, $V$, and a prefect
%sequence $A$, all having the same period $N$. We denote their DFTs by
%$\overline{A}=(\overline{a}_0,\overline{a}_1,\cdots,\overline{a}_{N-1})$,
%$\overline{U}=(\overline{u}_0,\overline{u}_1,\cdots,\overline{u}_{N-1})$,
%$\overline{V}=(\overline{v}_0,\overline{v}_1,\cdots,\overline{v}_{N-1})$,
%and the DFT of $\theta_{UV}(n)$ by $\Theta_{UV}(k)$. For the two
%modulated sequences $\widetilde{U}=U \circ A$ and
%$\widetilde{V}=V \circ A$, \textit{Lemma \ref{CC_DFT}} and the normalization
%$|\overline{a}_k|^2=1$, $\forall ~k$, imply that
%%$\Theta_{UV}(k)=\Theta_{\widetilde{U}\widetilde{V}}(k)$
%for every $k$,
%\begin{IEEEeqnarray}{rCl}
%\Theta_{\widetilde{U}\widetilde{V}}(k)= \overline{u}_k
%\overline{a}_k^{\ast} \overline{v}^{\ast}_k
%\overline{a}_k=|\overline{a}_k|^2 \overline{u}_k \overline{v}_k^{\ast} =
%\overline{u}_k \overline{v}_k^{\ast} = \Theta_{UV}(k)
%\end{IEEEeqnarray}
%Similarly, we can verify that
%$\Theta_{\widetilde{U}\widetilde{U}}(k)=\Theta_{UU}(k)$ and
%$\Theta_{\widetilde{V}\widetilde{V}}(k)$ $=\Theta_{VV}(k)$. Hence,
%we conclude that
\begin{lemma}
The AC and CC functions of a set of sequences are invariant (up to
a scaling factor) to filtering if the filtering sequence $v$ is a
perfect sequence. \label{mod_inv}
\end{lemma}

As will become clear in subsequent sections that this lemma makes
the filtering operator very useful in transforming a sequence set
into one with entries of the sequences taken from a desired constellation
while maintaining the correlation properties.

%In the subsequent discourse, approaches are proposed to produce
%ZCZ sequences.
%That is, the following constructions are limited .

\subsection{Basic Constructions}
\begin{definition}
A sequence $\{u(n)\}$ in an $(N,K)$ sequence set is said to have a
{\it subperiod} of $J$, where $J|N$, if it is also periodic with
period $J<N$, i.e., $u(n)=u(\ell J+n)$, for $0 \leq \ell <N/J$ and
$0\leq n < J$.
\end{definition}
Now note that \textit{Lemma \ref{CC_DFT}} implies
\begin{IEEEeqnarray}{rCl}
\theta_{uv}(\tau)
=\sum_{k=0}^{N-1}\Theta_{uv}(k)e^{\frac{j2\pi\tau k}{N}}
=\sum_{k=0}^{N-1}U(k)V^*(k)e^{\frac{j2\pi\tau k}{N}}~~~~
\label{transformCC}
\end{IEEEeqnarray}
where $\Theta_{uv}(k)=$DFT$\{\theta_{uv}(\tau)\}$. When $\{U(k)\}$
and $\{V(k)\}$, regarded as $N$-dimensional vectors, are
orthogonal, we have
\begin{eqnarray}
  \theta_{uv}(0)=\sum_{k=0}^{N-1}\Theta_{uv}(k)
  =\sum_{k=0}^{N-1}{U}(k)V^*(k) =0. \label{equation:cx_at_0}
\end{eqnarray}
If the sequence $\{\Theta_{uv}(k)\}$ has a subperiod of
$J=\frac{N}{m}$, then
%\begin{flushright}
\begin{IEEEeqnarray}{rCl}
\theta_{uv}(\tau)&=&\sum_{k=0}^{J-1}\Theta_{uv}(k)e^{\frac{j2\pi\tau k}{N}}
+\sum_{k=J}^{2J-1}\Theta_{uv}(k)e^{\frac{j2\pi\tau k}{N}}+ \cdots %\nonumber\\
+\sum_{k=(m-1)J}^{N-1}\Theta_{uv}(k)e^{\frac{j2\pi\tau k}{N}}%~~~~~~~~~~~~~~~~~~~~~~~
\nonumber\\
&=&\sum_{k=0}^{J-1}\Theta_{uv}(k)e^{\frac{j2\pi\tau k}{N}}
\left(1+e^{\frac{j2\pi\tau}{m}}+\cdots+e^{\frac{j2\pi\tau(m-1)}{m}}\right)\notag \label{equation:CC_time_freq}
\end{IEEEeqnarray}
%\end{flushright}
The identity
\begin{equation}
1+\alpha+\alpha^2+\cdots+\alpha^{m-1}=0,~\forall~ \alpha=W_{m}^{\tau},~|\tau|_m\neq 0
\label{eqn:root}
\end{equation}
then gives
\begin{lemma}
The CC function $\theta_{uv}(\tau)$ of two period-$N$ sequences
$\{u(n)\}$ and $\{v(n)\}$ is identical zero $\forall|\tau|_N\leq
T$ if the associated DFT vectors $\{U(k)\}$ and $\{V(k)\}$ are
orthogonal and their Hadamard product, $\{U(k)V^*(k)\}$, has a
subperiod of $J={N}/({T+1})$, where $T$ is a positive integer.
\end{lemma}

The recursive Kronecker construction of the Sylvester Hadamard
matrices (\ref{eq:recur_gen}) gives at least two sets of row
vectors (i.e., upper- and lower-half parts of ${\bf H}_{2^n}$)
that satisfy both the orthogonality and subperiodicity
requirements. This property still holds when we replace Sylvester
Hadamard matrices by other classes of Hadamard matrices produced
by a recursive Kronecker construction similar to
(\ref{eq:recur_gen}). Furthermore, as elements of a Hadamard
matrix have constant modulus, the AC of all sequences derived by
taking IDFT on rows of a Hadamard matrix is $0$ for all nonzero
correlation lags by \textit{Corollary \ref{perfect_seq}}. These
two observations suggest that ZCZ families can be obtained by
using proper subsets of row vectors from a Hadamard matrix. To
have a precise definition of ``proper subsets," we need
\begin{definition}
A regular $p$th-order $M$-partition on an $N\times N$ matrix
$\mathbf{H}$, where $N=M^n$, is the set of $m=N/K=M^p$ $K \times
N$ submatrices, each is formed by non-overlapping $K=M^{n-p}$
consecutive rows of $\mathbf{H}$. \label{def:m-partition}
\end{definition}
Proper subsets of row vectors that generate ZCZ families are
obtained by performing $p$th-order $M$-partition on the $n$th
Kronecker power of a Hadamard matrix, i.e.,
\begin{lemma}
Let $\mathbf{U}$ be a Hadamard matrix of order $M$ and
$\mathbf{H}$ be the Hadamard matrix of order $N$ generated by the
$n$th Kronecker power of $\mathbf{U}$, i.e.,
\begin{eqnarray}
  \mathbf{H}
  =[{\mathbf h}^T_0, {\mathbf h}^T_1,\cdots,{\mathbf
     h}^T_{N-1}]^T
  =\otimes^n \mathbf{U},
\label{eq:Hadarmard_Kron_power}
\end{eqnarray}
where $N=M^n$, $n\geq2$, and ${\mathbf h}_\ell$ is the $\ell$th
row\footnote{For convenience, all the column, row, and vector
elements' indices start with 0 instead of 1.} of $\mathbf{H}$. We
perform a regular $p$th-order $M$-partition on $\mathbf{H}$ to
obtain the $m=M^p$ submatrices
\begin{eqnarray}
\widetilde{{\mathbf H}}_i=[{\mathbf h}^T_{iK}, \cdots, {\mathbf
h}^T_{(i+1)K-1}]^T, ~i=0,1,\cdots,m-1.
%\stackrel{def}{=}[\tilde{\mathbf a}_{i,0}^T,
%\tilde{\mathbf a}_{i,1}^T, \cdots, \tilde{\mathbf a}_{i,K-1}^T],
\end{eqnarray}
Then, for each $i$, the set of $K$ length-$N$ sequences
$\mathbf{A}_i \stackrel{def}{=}\{ A_{i,0},A_{i,1},\cdots,
A_{i,K-1}\}$, where $A_{i,j}=\text{IDFT}\{\mathbf{h}_{iK+j}\}$, is
an $(N,K,m-1)$ ZCZ sequence family that achieves the upper bound
(\ref{eq:bound}). Furthermore, all member sequences in the family
are perfect sequences.
%If $K=M^\ell$, $\ell < n$ so that $m=M^{n-\ell}=M^p, p >0$ then the
%associated family achieves the upper bound (\ref{eq:bound}).
\label{partition_fundamental}
\end{lemma}
\begin{proof}
The matrix $\mathbf{H}$ can be expressed in the stacked form,
${\mathbf H}=\left[\widetilde{{\mathbf H}}_0^T, \widetilde{
{\mathbf H}}^T_1,\cdots,\widetilde{{\mathbf H}}^T_{m-1}\right]^T$,
where the submatrix $\widetilde{{\mathbf H}}_i$ is of the form
\[
\left[a_{i,0}\mathbf{B}, a_{i,1}\mathbf{B},\cdots,
a_{i,m-1}\mathbf{B}\right]
\]
where $a_{i,j}$'s have unit magnitudes and ${\bf B}=\otimes^{n-p}
\mathbf{U}$. It follows immediately that the Hadamard products of
two distinct rows of $\widetilde{{\mathbf H}}_i$ has a period of
$M^{n-p}=K$.
\end{proof}
The above construction gives ZCZ sequences of length $M^n$, $n\geq2$.
That the upper bound (\ref{eq:bound}) is achieved is a result
of our partition method described by {\it Definition
\ref{def:m-partition}}. The sequence length constraint can be
relaxed by using Kronecker construction of Hadamard matrices of
different orders. Using \textit{Lemma \ref{Hada_Hada_Hada}} and an
argument similar to that in deriving the above lemma, we obtain
\begin{theorem}
Let $\mathbf{H}$ be the $N \times N$ Hadamard matrix
\begin{IEEEeqnarray}{rCl}
  \mathbf{H}=[{\mathbf h}^T_0, {\mathbf h}^T_1,\cdots,{\mathbf
     h}^T_{N-1}]^T\stackrel{def}{=}\mathbf{U}_{n-1}\otimes\cdots\otimes\mathbf{U}_0
  \label{eq:GeneralizedHadarmard}
\end{IEEEeqnarray}
where $\mathbf{U}_k$, $k=0,1,\cdots,n-1$, are $M_k\times M_k$
(not necessarily distinct) Hadamard matrices and $N=\prod_{k=0}^{n-1}M_k$,
$n\geq2$. Partition $\mathbf{H}$ into $m=\frac{N}{K}$
submatrices of size $K\times N$,
\begin{eqnarray}
\widetilde{{\mathbf H}}_i=[{\mathbf h}^T_{iK}, \cdots, {\mathbf
h}^T_{(i+1)K-1}]^T, ~i=0,1,\cdots,m-1,
\label{eq:submtx}
\end{eqnarray}
each formed by non-overlapping $K=\prod_{k=0}^{n-p-1}M_k$
consecutive rows of $\mathbf{H}$ with $p>0$. Then, for each $i$,
the set of $K$ period-$N$ sequences $\mathbf{A}_i
\stackrel{def}{=}\{ A_{i,0},A_{i,1},\cdots, A_{i,K-1}\}$, where
$A_{i,j}=\text{IDFT}\{\mathbf{h}_{iK+j}\}$, is an $(N,K,m-1)$ ZCZ
sequence family that achieves the upper bound
(\ref{eq:bound})\footnote{Technically, the theorem is also valid
for $p=0$, as the resulting set has a ZCZ width 0. We will
implicitly ignore this trivial case and assume $p>0$ in the
subsequent discussion.}. \label{trans_fundamental}
\end{theorem}
%\begin{proof}
%See Appendix \ref{app:pf_trans_fundamental}.
%\end{proof}

Note that the recursive generation of Hadamard matrices defined by
(\ref{eq:recur_gen}) and (\ref{eq:Hadarmard_Kron_power}) are
special cases of (\ref{eq:GeneralizedHadarmard}), i.e., the above
theorem generalize \textit{Theorems 1} and \textit{2} of
\cite{isit2009}.

\subsection{Polyphase ZCZ Sequences}
The ZCZ sequences generated by the methods described above are not
necessary of constant modulus but can be converted into polyphase
sequences without altering the desired AC and CC properties by a
proper filtering process; see \textit{Definition \ref{def:mod}}
and \textit{Lemma \ref{mod_inv}}. To find the filtering perfect
sequences we need the following two properties.
\begin{lemma}%=\{u(n)\} u(n)\in
\cite{poly_perf_prop} Let $U$ be a length-$N$ polyphase perfect
sequence with entries drawn from the $N$-PSK constellation. Then
both IDFT$\{U\}$ and DFT$\{U\}$ are polyphase perfect sequences.
\label{poly_DFT_poly}
\end{lemma}

\begin{lemma}
\cite{poly_perf_prop_2} Let $L$ be a natural number and $N=L^2$.
Define the length-$N$ polyphase sequence $\{u(k)\}$ by
 \begin{equation}
  u(k_1L+k_2)=W_{L}^{\beta(k_2)k_1+r(k_2)},~0\leq k_1, k_2<L,
  \label{eq:Mow_poly_perf}
 \end{equation}
where $\{\beta(k_2): k_2=0,1,\cdots, L-1\}$ is a permutation of
$\{0,1,\cdots,L-1\}$, and $r(k_2)$ is a rational number depending
on $k_2$. Then the sequences, $\{u(k)\}$,
 \begin{IEEEeqnarray}{rCl}
  \{e^{j\theta_{k_2}}u(k_1L+k_2):0\leq\theta_{k_2}<2\pi,~%\nonumber\\&&
    0\leq k_1, k_2<L\}~~~~%\IEEEeqnarraynumspace
 \end{IEEEeqnarray}
and
 \begin{IEEEeqnarray}{rCl}
\{W_{L}^{\ell k_1}u(k_1L+k_2):0\leq k_1, k_2<L\},~ \text{for any
integer $\ell$,}\notag\\
  \label{eq:Mow_poly_perf_2}
 \end{IEEEeqnarray}
are all polyphase perfect sequences. \label{poly_DFT_poly2}
\end{lemma}
Based on the above results, we propose a transform domain
construction of polyphase ZCZ sequences as follows.
\begin{corollary}
Let $\mathbf{u}$ be a length-$N$ perfect sequence of the form
(\ref{eq:Mow_poly_perf}), $N=\prod_{k=0}^{n-1}M_k=L^{2}$ for some
$L$, and $\widetilde{{\mathbf H}}_i$ be the $i$th submatrix defined
by (\ref{eq:GeneralizedHadarmard}) and (\ref{eq:submtx}) using
$M_k$-DFT or $M_k$-IDFT matrices $\mb{U}_k$'s. Then
$\mathbf{C}_i=\big\{$IDFT$\{\mathbf{h}_{iK+n}\}\circ$IDFT$\{\mathbf{u}\}:
0\leq n \leq K-1\big\}$ is an $\left(N,K,\frac{N}{K}-1\right)$
bound-achieving polyphase ZCZ sequence set.
\label{poly_trans}
\end{corollary}
\begin{proof}
Since the entries in the $n$th row of $\widetilde{{\mathbf H}}_i$
render the general expression
 \begin{equation}
  [\mathbf{H}]_{iK+n,k_1L+k_2}\defeq h_{iK+n} (k_1L+k_2)=
  e^{j\theta_{k_2}(n)}W_{L}^{\ell(n)k_1}\nonumber
 \end{equation}
for $0\leq k_1,k_2<L$, where $\ell(n)\in\mathbb{Z}$ (integers) and
$0\leq\theta_{k_2}(n)<2\pi$, the products $h_{iK+n}(k)u^*(k)$ are
of the forms (\ref{eq:Mow_poly_perf})--(\ref{eq:Mow_poly_perf_2})
and are integer powers of $W_N$. \textit{Lemmas
\ref{poly_DFT_poly}} and \textit{\ref{poly_DFT_poly2}} imply that
the sequence
\begin{IEEEeqnarray}{rCl}
C_{i,n}(k)&=&\text{IDFT}\{h_{iK+n}(k)\}\circ\text{IDFT}\{u(k)\}\nonumber\\
&=& \text{IDFT}\{h_{iK+n}(k)u^*(k)\}\nonumber
\end{IEEEeqnarray}
has polyphase entries. %, where $|k|_L\stackrel{def}{=}k$ mod $L$.
Invoking \textit{Theorem \ref{trans_fundamental}} and
\textit{Lemma \ref{mod_inv}}, we conclude that $\{C_{i,n}:0\leq
n<K\}$ is an $(N,K,\frac{N}{K}-1)$ polyphase ZCZ family.
\end{proof}

\begin{remark} (Polyphase constraint and sequence length selection)
\textit{Theorem \ref{trans_fundamental}} provides a general transform
domain approach using Hadamard matrices to construct bound-achieving
sets of arbitrary nonprime length ZCZ sequences. In contrast,
{\it Corollary \ref{poly_trans}} focuses on the generation of polyphase
sequences and can be regarded as an extension of a special case of the
former. The polyphase requirement is satisfied by invoking an additional
filtering operation and the use of special Hadamard matrices; see {\it
Example 1} in the ensuing subsection. As a result, the choice of the
sequence length is limited to perfect squares ($N=L^2$).
\end{remark}
\begin{remark} (Nested structure)
Every $K\times N$ submatrix $\widetilde{{\mathbf H}}_i$ can be
further partitioned into $K/K'=\prod_{k=n-p'}^{n-p-1}M_k$
submatrices of size $K'\times N$, where $p<p'<n$,
$K=\prod_{k=0}^{n-p-1}M_k$, and $K'=\prod_{k=0}^{n-p'-1}M_k$ so
that each submatrix can be used to construct an
$(N,K',\frac{N}{K'}-1)$ ZCZ sequence set $\mathbf{C}_i^j$ with
larger ZCZ width and $\bigcup_{j=0}^{\frac{K}{K'}-1} {\bf
C}_i^j={\bf C}_i$. This partition can be done in a nested manner,
i.e., each subset can be further decomposed to render even smaller
sequence subsets or $\widetilde{{\bf H}}_i$ can be merged with
proper neighboring submatrices to construct a larger set.
\end{remark}
\begin{remark} (Tradeoff between AC and CC)
The identity (\ref{eqn:root}) actually gives a stronger CC
property than what is specified by the ZCZ width; it implies that
the CC values are identically zero except at $\tau=s(T+1)$,
$s\in\mathbb{Z}$. This is still weaker than the constructions
of \cite{Transform} and \cite{Zak_trans} which yield perfect
(zero) CC at all lags. Perfect CC is achieved by requiring that
each transform domain sequence has sparse nonzero elements and
support (set of the nonzero coordinates) disjoint from the supports
of all other transform domain sequences. Nevertheless, their AC
functions are not as good as ours as all the sequences constructed
by our approach are perfect sequences.
\end{remark}
\begin{remark} (Tradeoff between sequence length and alphabet size)
Tsai's approach \cite{Transform} is more flexible in the choice of
sequence length but requires a very large constellation for
elements of the sequences. Our approach, on the other hand,
requires the smallest constellation and is more flexible than
\cite{Zak_trans} in selecting the sequence length $N$.
\end{remark}

We summarize various parameter constraints for our approach,
\cite{Transform}, and \cite{Zak_trans} in Table \ref{tab:trans_comp}.

\renewcommand{\arraystretch}{1.6}
\begin{table}[t]
 \caption{Transform domain-based polyphase ZCZ sequence sets}
 \centering
 \tabcolsep 0.03in
 \label{tab:trans_comp}
 \footnotesize{
 \begin{tabular}{|C{1in}|C{1in}|C{1in}|C{1in}|
                  C{1in}|}
 \hline & Tsai \cite{Transform} & \multicolumn{2}{c|}{Brodzik \cite{Zak_trans}} & \textit{Corollary
 \ref{poly_trans}}\\
  \hline Sequence length $N$ & $n_1n_2$ & \multicolumn{2}{c|}{$L^3$} &
  $\prod_{k=0}^{n-1}M_k=L^2$\\ \hline
 Set size $K$ & $n_2$ & \multicolumn{2}{c|}{$L$} &
  $\prod_{k=0}^{n-p-1}M_k$\vspace{.2em} \\ \hline
 ZCZ width $T$ & $n_1-1$ & $L^2-1$, $L$ prime & $L-1$, $L$ nonprime & $\frac{N}{K}-1$ %\multirow{2}{*}{}
  %$\prod_{k=n-p}^{n-1}M_k-1$\vspace{.2em}
  \\ \hline
% Constraints on $n_2$ & ~ &
%  $\gcd(n_1,n_2)=1$ & $n_1|n_2$ or $n_2|n_1$ &
%  $\gcd(n_1,n_2)$ $=1$  \\ \hline
 Upper-bound (\ref{eq:bound}) achieved? & Yes & Yes & No & Yes  \\ \hline
 Perfect sequence used & Length-$n_1$, $n_P$-phase & \multicolumn{2}{c|}{No explicit use of perfect sequences} & Length-$N$ \\
 \hline
 Alphabet size & lcm$(N,n_P)$ & \multicolumn{2}{c|}{$N$} &
  $N$ \\ \hline
%  1 & 2 & 3 & 4 & 6 & 7 & 8 & 9 \\ \hline
 \end{tabular}}
\end{table}
\renewcommand{\arraystretch}{1}

\subsection{ZCZ Sequence Sets Generated by Transform Domain
Approach} \label{section:trans_ex} In this subsection, we present
some construction examples using the proposed transform domain
method. All ZCZ sequences obtained are perfect sequences. To
minimize the number of notations, we use $C_i$ and $A_i$ to denote
sequences generated by the methods of \textit{Corollary
\ref{poly_trans}} and \textit{Theorem \ref{trans_fundamental}},
respectively. The same notation may refer to different sequences
in different examples when there is no danger of ambiguity.

\begin{example}
(Use of three DFT matrices of unequal dimensions)~ Partitioning
the Hadamard matrix $\mathbf{H}=\mathbf{F}_6\otimes\mathbf{F}_3
\otimes\mathbf{F}_2$ into submatrices $\widetilde{\mathbf{H}}_{0},
\widetilde{\mathbf{H}}_{1}, \cdots,\widetilde{\mathbf{H}}_{17}$
and performing IDFT on the rows of $\widetilde{\mathbf{H}}_{10}$,
we obtain two sequences
\begin{IEEEeqnarray*}{rCl}
  A_{0}=(
0   0   0   W_{12}^{21} 0   0   0   0   0   W_{12}^{7} 0 00 0   0    W_{12}^{5} 0   0   0   0 0 W_{12}^{15} 0 00    0   0  W_{12}^{1} 0   00  0 0 W_{12}^{23}0 0),&&\\
  A_{1}=(0  0   0   W_{12}^{15} 0   0   0   0   0   W_{12}^{1}  0    00 0   0   W_{12}^{23} 0   0   0   0 0 W_{12}^{21} 0   00  0    0   W_{12}^{7}   0   00  0 0  W_{12}^{5} 0   0).&&
\end{IEEEeqnarray*}
To convert them into ones
with constant moduli we filter them by the perfect polyphase
sequence \cite{Transform}
\begin{IEEEeqnarray}{rCl}
U_{36}=(W_6^0 W_6^0 W_6^0 W_6^0 W_6^0 W_6^0 W_6^0 W_6^5 W_6^4 W_6^3 W_6^2 W_6^1W_6^0 W_6^4 W_6^2 W_6^0 W_6^4 W_6^2 ~~~~&&\nonumber\\
W_6^0 W_6^3 W_6^0 W_6^3 W_6^0 W_6^3%~~~~&&\nonumber\\
W_6^0 W_6^2 W_6^4 W_6^0 W_6^2 W_6^4 W_6^0 W_6^1 W_6^2 W_6^3 W_6^4
W_6^5)~~~&& \label{eq:L36_perf_poly}
\end{IEEEeqnarray}
which satisfies (\ref{eq:Mow_poly_perf}). The resulting
$(36,2,17)$ bound-achieving ZCZ sequence set consists of
\begin{IEEEeqnarray*}{rCl}
C_0=A_{0}\circ U_{36}=\:
&&(W_{12}^{11}  W_{12}^{8}  W_{12}^{9}  W_{12}^{0}  W_{12}^{3}   W_{12}^{2} W_{12}^{5}  W_{12}^{4}  W_{12}^{7}%\nonumber\\&&~
W_{12}^{0}   W_{12}^{5}  W_{12}^{6}  W_{12}^{11} W_{12}^{12}  W_{12}^{5} W_{12}^{0}  W_{12}^{7}  W_{12}^{10}\nonumber\\
&&~W_{12}^{5}   W_{12}^{8}  W_{12}^{3}  W_{12}^{0}  W_{12}^{9}   W_{12}^{2} W_{12}^{11} W_{12}^{4}  W_{12}^{1}%\nonumber\\&&~
W_{12}^{0}   W_{12}^{11} W_{12}^{6}  W_{12}^{5}  W_{12}^{12}  W_{12}^{11}    W_{12}^{0}  W_{12}^{1}  W_{12}^{10}),\nonumber\\
C_1=A_{1}\circ U_{36}=\:
&&(W_{12}^{5}   W_{12}^{8}  W_{12}^{3}  W_{12}^{0}  W_{12}^{9}   W_{12}^{2} W_{12}^{11} W_{12}^{4}  W_{12}^{1}%\nonumber\\&&~
W_{12}^{0}   W_{12}^{11} W_{12}^{6}  W_{12}^{5}  W_{12}^{12}  W_{12}^{11}    W_{12}^{0}  W_{12}^{1}  W_{12}^{10}\nonumber\\
&&~W_{12}^{11}  W_{12}^{8}  W_{12}^{9}  W_{12}^{0}  W_{12}^{3}   W_{12}^{2} W_{12}^{5}  W_{12}^{4}  W_{12}^{7}%\nonumber\\&&~
W_{12}^{0}   W_{12}^{5}  W_{12}^{6}  W_{12}^{11} W_{12}^{12}  W_{12}^{5} W_{12}^{0}  W_{12}^{7}  W_{12}^{10}).
\end{IEEEeqnarray*}
If instead we take IDFT on the rows of the first submatrix
$\widetilde{\mathbf{G}}_{0}$ of
$\mathbf{G}=[\widetilde{\mathbf{G}}_{0}^T,
\widetilde{\mathbf{G}}_{1}^T, \cdots,
\widetilde{\mathbf{G}}_{11}^T]^T = \mathbf{F}_2\otimes\mathbf{F}_6
\otimes\mathbf{F}_3$ and filter the resulting sequences
$\{A_{0},A_{1},A_{2}\}$ through (\ref{eq:L36_perf_poly}), we
obtain the bound-achieving $(36,3,11)$ set:
\begin{IEEEeqnarray*}{rCl}
C_0=A_{0}\circ U_{36}=(W_6^0 W_6^1 W_6^2 W_6^3 W_6^4 W_6^5 W_6^0 W_6^2 W_6^4
%~~&&\nonumber\\
W_6^0 W_6^2 W_6^4 W_6^0 W_6^3 W_6^0 W_6^3 W_6^0 W_6^3~~&&\nonumber\\
W_6^0 W_6^4 W_6^2 W_6^0 W_6^4 W_6^2 W_6^0 W_6^5 W_6^4%~~&&\nonumber\\
W_6^3 W_6^2 W_6^1 W_6^0 W_6^0 W_6^0 W_6^0 W_6^0 W_6^0),&&\nonumber\\
C_1=A_{1}\circ U_{36}=(W_6^0 W_6^5 W_6^4 W_6^3 W_6^2 W_6^1 W_6^0 W_6^0 W_6^0
%~~&&\nonumber\\
W_6^0 W_6^0 W_6^0 W_6^0 W_6^1 W_6^2 W_6^3 W_6^4 W_6^5~~&&\nonumber\\
W_6^0 W_6^2 W_6^4 W_6^0 W_6^2 W_6^4 W_6^0 W_6^3 W_6^0%~~&&\nonumber\\
W_6^3 W_6^0 W_6^3 W_6^0 W_6^4 W_6^2 W_6^0 W_6^4 W_6^2),&&\nonumber\\
C_2=A_{2}\circ U_{36}=(W_6^0 W_6^3 W_6^0 W_6^3 W_6^0 W_6^3 W_6^0 W_6^4 W_6^2
%~~&&\nonumber\\
W_6^0 W_6^4 W_6^2 W_6^0 W_6^5 W_6^4 W_6^3 W_6^2 W_6^1~~&&\nonumber\\
W_6^0 W_6^0 W_6^0 W_6^0 W_6^0 W_6^0 W_6^0 W_6^1 W_6^2%~~&&\nonumber\\
W_6^3 W_6^4 W_6^5 W_6^0 W_6^2 W_6^4 W_6^0 W_6^2 W_6^4).&&
\end{IEEEeqnarray*}
\end{example}

\begin{example} (Construction based on Kronecker power of a DFT
matrix)~ Let $\mathbf{H}=\mathbf{F}_3\otimes\mathbf{F}_3
\otimes\mathbf{F}_3\otimes\mathbf{F}_3$ and denote by
$\widetilde{{\mathbf H}}_0,\widetilde{{\mathbf
H}}_1,\cdots,\widetilde{{\mathbf H}}_{26}$ the submatrices
obtained by performing regular $3$rd-order $3$-partition on
$\mathbf{H}$. Choosing $\widetilde{\mathbf H}_2$ and performing
IDFT on its rows, we obtain sequences $\{A_{0},A_{1},A_{2}\}$.
Filtering them by polyphase perfect sequence
\begin{IEEEeqnarray*}{rCl}
U_{81}=(W_9^0   W_9^0   W_9^0   W_9^0   W_9^0   W_9^0   W_9^0
W_9^0 W_9^0   W_9^0   W_9^8   W_9^7   W_9^6 W_9^5
W_9^4
W_9^3 W_9^2 W_9^1~~&&\nonumber\\   W_9^0   W_9^7   W_9^5 W_9^3W_9^1   W_9^8 W_9^6   W_9^4 W_9^2   W_9^0   W_9^6 W_9^3   W_9^0
W_9^6 W_9^3 W_9^0 W_9^6 W_9^3~~&&\nonumber\\
   W_9^0 W_9^5   W_9^1
W_9^6 W_9^2 W_9^7   W_9^3   W_9^8  W_9^4 W_9^0 W_9^4
W_9^8 W_9^3   W_9^7   W_9^2   W_9^6 W_9^1 W_9^5~~&&\nonumber\\
W_9^0 W_9^3 W_9^6 W_9^0   W_9^3   W_9^6   W_9^0
W_9^3 W_9^6   W_9^0 W_9^2 W_9^4 W_9^6   W_9^8 W_9^1
W_9^3 W_9^5 W_9^7~~&&\nonumber\\ W_9^0   W_9^1 W_9^2 W_9^3   W_9^4
W_9^5 W_9^6 W_9^7 W_9^8
),&&%\nonumber\\
\end{IEEEeqnarray*}
we obtain
%\cite{Transform} for instance,
%The resulting $(81,3,26)$ bound-achieving ZCZ sequence set
%consists of
\begin{eqnarray*}
C_0=A_{0}\circ U_{81}=( W_9^0   W_9^1   W_9^2   W_9^6   W_9^7
W_9^8 W_9^3   W_9^4 W_9^5   W_9^0   W_9^2   W_9^4
W_9^0 W_9^2 W_9^4 W_9^0 W_9^2   W_9^4 ~~&&\nonumber\\  W_9^0   W_9^3 W_9^6 W_9^3
W_9^6 W_9^0 W_9^6   W_9^0   W_9^3   W_9^0   W_9^4 W_9^8 W_9^6
W_9^1 W_9^5 W_9^3   W_9^7   W_9^2 ~~&&\nonumber\\W_9^0 W_9^5 W_9^1
W_9^0 W_9^5 W_9^1 W_9^0   W_9^5   W_9^1 W_9^0 W_9^6
W_9^3 W_9^3 W_9^0 W_9^6 W_9^6 W_9^3
W_9^0~~&&\nonumber\\ W_9^0
W_9^7 W_9^5 W_9^6 W_9^4 W_9^2 W_9^3 W_9^1 W_9^8 W_9^0 W_9^8 W_9^7
W_9^0 W_9^8 W_9^7 W_9^0 W_9^8 W_9^7~~&&\nonumber\\ W_9^0 W_9^0 W_9^0 W_9^3 W_9^3
W_9^3 W_9^6 W_9^6 W_9^6),&&\nonumber\\
C_1=A_{1}\circ U_{81}=(W_9^0   W_9^7   W_9^5   W_9^6 W_9^4 W_9^2
W_9^3 W_9^1 W_9^8 W_9^0 W_9^8   W_9^7 W_9^0
W_9^8 W_9^7 W_9^0 W_9^8 W_9^7 ~~&&\nonumber\\W_9^0 W_9^0 W_9^0 W_9^3W_9^3 W_9^3 W_9^6 W_9^6 W_9^6 W_9^0 W_9^1 W_9^2 W_9^6 W_9^7
W_9^8 W_9^3   W_9^4 W_9^5 ~~&&\nonumber\\
W_9^0  W_9^2 W_9^4 W_9^0
W_9^2   W_9^4 W_9^0   W_9^2 W_9^4   W_9^0 W_9^3
W_9^6 W_9^3 W_9^6   W_9^0 W_9^6 W_9^0 W_9^3~~&&\nonumber\\
 W_9^0W_9^4 W_9^8 W_9^6 W_9^1   W_9^5 W_9^3 W_9^7 W_9^2   W_9^0 W_9^5
W_9^1 W_9^0 W_9^5 W_9^1 W_9^0
W_9^5 W_9^1 ~~&&\nonumber\\W_9^0   W_9^6   W_9^3 W_9^3 W_9^0
W_9^6 W_9^6 W_9^3 W_9^0),&&\nonumber\\
C_2=A_{2}\circ U_{81}=(W_9^0   W_9^4   W_9^8   W_9^6   W_9^1 W_9^5
W_9^3   W_9^7 W_9^2   W_9^0   W_9^5   W_9^1
W_9^0 W_9^5 W_9^1 W_9^0 W_9^5   W_9^1 ~~&&\nonumber\\  W_9^0   W_9^6   W_9^3   W_9^3
W_9^0   W_9^6 W_9^6   W_9^3   W_9^0   W_9^0   W_9^7   W_9^5 W_9^6
W_9^4 W_9^2  W_9^3   W_9^1   W_9^8 ~~&&\nonumber\\  W_9^0   W_9^8
W_9^7   W_9^0 W_9^8   W_9^7   W_9^0   W_9^8  W_9^7
W_9^0 W_9^0 W_9^0 W_9^3 W_9^3   W_9^3   W_9^6   W_9^6   W_9^6~~&&\nonumber\\
W_9^0   W_9^1 W_9^2   W_9^6 W_9^7   W_9^8   W_9^3
W_9^4 W_9^5 W_9^0 W_9^2 W_9^4   W_9^0 W_9^2   W_9^4
W_9^0 W_9^2 W_9^4~~&&\nonumber\\ W_9^0   W_9^3 W_9^6   W_9^3 W_9^6
W_9^0 W_9^6 W_9^0 W_9^3)~&&\nonumber%\\
\end{eqnarray*}
which form an $(81,3,26)$ ZCZ sequence set that satisfies
(\ref{eq:bound}).
\end{example}

%\subsubsection{}
\begin{example}
(Quadriphase sequences derived from a Sylvester Hadamard matrix)~
Partition the Sylvester Hadamard matrix $\mathbf{H}_{16}$ into
four submatrices, $\widetilde{\mathbf{H}}_{0},
\widetilde{\mathbf{H}}_{1}, \widetilde{\mathbf{H}}_{2},
\widetilde{\mathbf{H}}_{3}$, and select the first submatrix,
$\widetilde{\mathbf{H}}_{0}=[\mathbf{h}_{0}^T, \mathbf{h}_{1}^T,
\mathbf{h}_{2}^T,\mathbf{h}_{3}^T]^{T}$. Filtering the IDFT of
$\mathbf{h}_{i}$ by
\begin{eqnarray}
U_{16}=(W_4^0W_4^0W_4^0W_4^0W_4^0W_4^3W_4^2W_4^1%~~&&\nonumber\\
   W_4^0W_4^2W_4^0W_4^2W_4^0W_4^1W_4^2W_4^3),&&
\label{eq:qpsk_perfect_seq}
\end{eqnarray}
for each $i$, we have
\begin{eqnarray*}
C_0=(W_4^0W_4^1W_4^2W_4^3W_4^0W_4^2W_4^0W_4^2%~~&&\nonumber\\
W_4^0W_4^3W_4^2W_4^1W_4^0W_4^0W_4^0W_4^0),&&\nonumber\\
C_1=(W_4^0W_4^3W_4^2W_4^1W_4^0W_4^0W_4^0W_4^0%~~&&\nonumber\\
W_4^0W_4^1W_4^2W_4^3W_4^0W_4^2W_4^0W_4^2),&&\nonumber\\
%\end{eqnarray}
%\begin{eqnarray}
C_2=(W_4^0W_4^1W_4^0W_4^1W_4^0W_4^2W_4^2W_4^0%~~&&\nonumber\\
W_4^0W_4^3W_4^0W_4^3W_4^0W_4^0W_4^2W_4^2),&&\nonumber\\
C_3 =(W_4^0W_4^3W_4^0W_4^3W_4^0W_4^0W_4^2W_4^2%~~&&\nonumber\\
W_4^0W_4^1W_4^0W_4^1W_4^0W_4^2W_4^2W_4^0),&&
%\label{example:polyphase}
\end{eqnarray*}
a quadriphase $(16,4,3)$ ZCZ sequence family that satisfies
(\ref{eq:bound}).

Note that if a $3$rd-order $2$-partition is used instead, we have
a set of only two sequences but with a larger ZCZ width, i.e., we
obtain a quadriphase $(16,2,7)$ ZCZ sequence set consisting of
$\{A_{0}\circ U_{16},A_{1}\circ U_{16}\}$ or $\{A_{2}\circ
U_{16},A_{3}\circ U_{16}\}$.
\end{example}

\section{Direct Synthesis Method}
\label{section:direct}
\subsection{Preliminaries}
We now present an alternate approach which is capable of
generating ZCZ sequences of arbitrary nonprime periods.
\begin{definition}
A binary (0- and 1-valued) sequence of period $N$ which satisfies
the ZCZ width constraint $T$ on its AC function is called a basic
$(N,T)$ sequence.
\end{definition}
A basic sequence can be obtained by the simple rule given in
\begin{lemma}
A binary sequence $B=(b_0, b_1,\cdots, b_{N-1})$, $b_i \in
\{0,1\}$, is a basic $(N,T)$ sequence if the minimum run length of
$0$'s is $T$ (in the circular sense), where a run refers to a
string of identical symbols and $T$ is also called the
\textit{minimum spacing} of $B$. \label{basic_seq}
\end{lemma}
\subsection{Synthesis Process}
Two new operations are needed.
\begin{definition}
A basic $(N,T)$ sequence $B$ with Hamming weight $w_H(B)$ can be
expressed as the sum (via component-wise addition) of $M$
length-$N$ binary sequences, $\{B_i\}_{i=0}^{M-1}$, with disjoint
nonempty supports so that $\sum_{i=0}^{M-1} w_H(B_i)=w_H(B)$ and
$w_H(B_i)\geq 1$. The sequence set $\{B_i\}_{i=0}^{M-1}$ is said to be
an {\it orthogonal tone decomposition} of $B$.
\label{def:ortho_tone_decomp}
\end{definition}
It is trivial to see that $\{B_i\}_{i=0}^{M-1}$ is a binary $(N,M,T)$
ZCZ sequence family and each $B_i$ is a basic $(N, T_i)$ sequence
with $T_i \geq T$.
%%\begin{definition}
%%Let $V=(v(0),v(1),\cdots,v(N-1))$ be a length-$N$ binary sequence
%%with Hamming weight $w_H(V)=k$, $s_V(m)=$ the $m$th nonzero
%%coordinate of the sequence $V$ and $\mathbf{U}=[u_{ij}]$, a $k'
%%\times k$ matrix. The $V$-upsampled matrix of $\mathbf{U}$ is the
%%$k' \times N$ matrix $\mathbf{P}=[p_{ij}]$
%%defined by
%%\begin{eqnarray}
%% p_{ij}= \left\{
%% \begin{array}{ll}
%% u_{im}, & j=s_V(m),~m=0,1,\cdots,k-1;\\
%%%   \     &  \ \ \ ~~\mbox{at the leftmost) nonzero entry of $V$}\\
%%%   \     & \mbox{where $0\leq m < k'$;} \\
%% 0, & \mbox{otherwise}. \\
%% \end{array}
%% \right.
%%%\nonumber
%%\end{eqnarray}
%%We denote the above row-wise nonuniform upsampling operation on
%%$\mathbf{U}$ by $\mathbf{P}=\mathbf{U} \vartriangle V$.
%%\label{def:nonuni_upsamp}
%%\end{definition}
\begin{definition}
Let $V=(v(0),v(1),\cdots,v(N-1))$ be a length-$N$ binary sequence
with Hamming weight $w_H(V)=k$ and $\mathbf{U}=[u_{ij}]$ be any
matrix having $k$ columns and arbitrary number of rows $k'$. The
$V$-upsampled matrix of $\mathbf{U}$ is the $k' \times N$ matrix
$\mathbf{P}=[p_{ij}]$ defined by
\begin{eqnarray}
 p_{ij}= \left\{
 \begin{array}{ll}
 u_{im}, & j=s_V(m),~m=0,1,\cdots,k-1;\\
 0, & \mbox{otherwise}, \\
 \end{array}
 \right.
\end{eqnarray}
where $s_V(m)=$ the coordinate of sequence $V$'s $m$th nonzero
entry. We denote the above row-wise nonuniform upsampling
operation on $\mathbf{U}$ by $\mathbf{P}=\mathbf{U} \vartriangle
V$. \label{def:nonuni_upsamp}
\end{definition}

Obviously, the nonzero entries in all rows of the matrix
$\mathbf{P}=\mathbf{U} \vartriangle V$ are in the same positions.
Hence if $V$ is an $(N,T)$ basic sequence constructed by the
procedure described in \textit{Lemma \ref{basic_seq}}, then each
row has the same minimum spacing $T$ and all CC (including AC)
values are zero at $0<\tau\leq T$. Values of all CC functions at
$\tau=0$ are zero when $\mathbf{U}$ is unitary in which case rows
of $\mathbf{P}$ all have ZCZ width $T$. Invoking \textit{Lemma
\ref{mod_inv}}, we have
\begin{lemma}
Let $B$ be a basic $(N,T)$ sequence with $w_H(B)=K$, ${\mathbf
B}\defeq\{B_i\}_{i=0}^{M-1}$ be an orthogonal tone decomposition
of $B$, $w_H(B_i)=k_i$, and ${\bf U}_i, 0\leq i < M$ be $k_i
\times k_i$ unitary matrices (not necessarily distinct). Then for
each $i$, the rows of nonuniform upsampled matrix $\mathbf{P}_i=
\mathbf{U}_i \vartriangle B_i$ constitute an $(N, K_i, T_i)$ ZCZ
sequence family, where $T_i \geq T$ is the minimum spacing of
$B_i$. Moreover, the rows of all $\mathbf{P}_i$'s constitute an
$(N,K,T)$ ZCZ sequence set.
%Furthermore, filtering
%each row of $\mathbf{P}_i$'s by a perfect sequence of length $N$,
%we obtain another $(N, K, T)$ family.
\label{direct_syn}
\end{lemma}

\subsection{Polyphase ZCZ Sequences}
\label{section:direct_poly} The above process does not guarantee a
constant modulus constellation for the entries of the generated
sequences. We need a special class of basic sequences and a
suitable perfect sequence to generates polyphase sequence
families.
\begin{theorem}
Let $A'=\{a'_n\}$ be a length-$N'$ perfect $N_{A'}$-PSK sequence,
where $2 \leq N_{A'} \leq 2N'$ and $A$ be the perfect sequence of
length $N= N_rN'$ derived from $N_r$-fold upsampling on $A'$. An
$(N,N_r, N'-1)$ or $(N,N_r, N'-2)$ ZCZ $\ell$-PSK sequence family,
where $\ell=\text{lcm}(N_{A'},N_r)$, can be obtained by filtering
the rows of $\mathbf{P}=\mathbf{F}_{N_r}\vartriangle B$ by $A$,
where $B=(b_0,b_1, \cdots,b_{N-1})$ is the weight-$N_r$ basic
sequence defined by
\begin{eqnarray}
b_i= \left\{
\begin{array}{ll}
1, &  i=k  N',~k=0,1,\cdots, N_r-1; \\
0, &  \mbox{otherwise},
\end{array}
\right. \label{eq:direct_prime}
\end{eqnarray}
if $N_r$ and $ N'$ are relatively prime, or by
\begin{eqnarray} b_i= \left\{
\begin{array}{ll}
1, &  i=k N', ~k=0,1, \cdots, \frac{L_0}{ N'}-1, ~\mbox{or} \\
   &  i=\ell L_0+\left(\frac{N}{L_0}-\ell\right)+k N',~\mbox{where}\\
   &  \ell=1,2,\cdots,\frac{N}{L_0}-1,\\
   &  k=0,1,\cdots, \frac{L_0}{ N'}-1;\\
0, &  \mbox{otherwise},
\end{array}
\right. \label{eq:direct_notprime}
\end{eqnarray}
if gcd$(N_r, N')\neq 1$, where $L_0=\text{lcm}(N_r,N')$.
\label{poly_by_direct}
\end{theorem}
\begin{proof}
See Appendix A.
\end{proof}

\subsection{Properties, Constraints, and Comparisons}
The following three properties about the approach described above
are easily verifiable.
\begin{remark}(Parameter relations)
For a fixed $N$ and $K=N_r$, ZCZ sequence families generated by
(\ref{eq:direct_prime}) achieve the upper bound (\ref{eq:bound})
and those generated from (\ref{eq:direct_notprime}) satisfy the
relation $K(T+1) = N-N_r$.
\end{remark}
\begin{remark} (Nested-like and inter-set properties)
\label{prop:shif_new} The construction described in {\it Lemma
\ref{direct_syn}} results in a nested-like structure similar to
that of {\it Remark 2}. Instead of decomposing a Hadamard matrix,
we decompose a basic sequence of minimum spacing $S$ into several
basic sequences of minimum spacing $S' \geq S$ and use the latter
basic sequences to construct sequence sets whose union constitutes
a larger ZCZ set with a ZCZ width smaller than that of individual
subset; see the second part of {\it Example 7}.

The construction of {\it Theorem \ref{poly_by_direct}} needs a
special choice of the Hadamard matrix and basic sequence used
because of the polyphase requirement. But as a special case of
{\it Lemma \ref{direct_syn}}, it still preserve the nested-like
structure. In fact, the basic sequences defined by
(\ref{eq:direct_prime}) and (\ref{eq:direct_notprime}) can be
cyclically shifted to generate distinct polyphase ZCZ sequence
families with the same $(N,K,T)$. The zero CC zone width between a
sequence from the set based on $B$ and one from the set based on a
circularly-shifted version of $B$ is determined by the CC function
of the two basic sequences used. If, for instance,
$\mathbf{C}_0=\{C_{0,0},C_{0,1},\cdots,C_{0,K-1}\}$ and
$\mathbf{C}_1=\{C_{1,0},C_{1,1},\cdots,C_{1,K-1}\}$ are derived
from basic sequence $B(n)$ and $B'(n)=B(|n-n'|_N)$, respectively,
then $\theta_{C_{0,i}C_{1,\ell}}(\tau)=0$, $\forall~ i, \ell$ and
$|\tau|_N\leq T'$, where $T'<T$ is the zero-CC zone width of
$\theta_{B,B'}(\tau)$. As a result, the set ${\bf C}_0 \cup {\bf
C}_1$ has the ZCZ width $T'<T$; see \textit{Example
\ref{ex:subs_direct}} in the next subsection.
\end{remark}
\begin{remark} (Binary sequences) \label{prop:Hada_mtx_L}
To generate binary ZCZ sequences one has to use binary Hadamard
matrices, which exist for $N_r=2^\ell$, $12\times 2^\ell$, or
$20\times 2^\ell$ \cite{Hada_ref}, to replace the $N_r$-DFT
matrix, ${\bf F}_{N_r}$, in constructing ${\bf P}$ and reduce the
required alphabet size to just $\text{lcm}(N_{A'},2)=2$; see
\textit{Examples \ref{ex:BiDirect1}} and
\textit{\ref{ex:BiDirect2}}.
\end{remark}

The parameter selection constraints and related properties for our
and some related existing methods are given in Table
\ref{tab:perf_invol}. We provide more comparisons in the following
remarks.
\begin{remark}
{\it Theorem \ref{poly_by_direct}} does not explicitly mention any
restriction on the alphabet size. As these constructions need to
use a length-$N'$ perfect sequence and $N_r\times N_r$ Hadamard
matrices, which do not always exist for all lengths ($N'$), matrix
dimension ($N_r$) and all constellation sizes ($N_{A'}$), the ZCZ
width, sequence length, and family size are thus implicitly
constrained by the alphabet size.
\end{remark}
\begin{remark}
Tang \textit{et al.} \cite{interleaving} classifies the ZCZ
sequences construction methods into two major categories, i.e., i)
those based on complementary sets and ii) those derived from
perfect sequences. Our approach belongs to the latter category and
generates sequences with length $N=n_1 n_2$, where $n_1$ is the
length of a perfect sequences. The constructions proposed in
\cite{Polyphase}--\cite{GCL} have similar
constraints on the sequence length $N$ and those mentioned in the
next three remarks.
\end{remark}
\begin{remark}
In \cite{Polyphase}, an $(N,k,(n_1-2)k^{\ell-1})$ set is
constructed by using a length-$n_1$ perfect sequence, where
$n_1=kt$, $k\leq n_1$, but $n_2$ must be of the form $k^\ell$,
$\ell>1$. The interleaving scheme \cite{interleaving} requires
that either i) gcd$(n_1,n_2)=1$ or ii) $n_1|n_2$ or $n_2|n_1$ to
generate an $(N,n_2,n_1-1)$ or $(N,n_2,n_1-2)$ ZCZ family. The
length constraints in i) is similar to that for the construction
(\ref{eq:direct_prime}) while ii) leads to ZCZ families of the
same parameters as those by the construction
(\ref{eq:direct_notprime}) except that the latter is only
constrained by gcd$(n_1,n_2)\neq 1$.
\end{remark}
\begin{remark}
A length-$N$ ($N=n_1n_2$) Frank-Chu perfect sequence is used in
\cite{GCL} to generate an $(N,n_2,n_1-1)$ family. This method also
calls for the use of an $n_2 \times n_2$ DFT or binary Hadamard
matrix. However, for the case when $n_1$ is a perfect square and a
DFT (or binary Hadamard) matrix is used, our approach needs an
alphabet of size lcm$(n_2,\sqrt{n_1})$ or lcm$(2,\sqrt{n_1})$
instead of lcm$(n_2,n_1)$, lcm$(n_2,2n_1)$ or lcm$(2,n_1)$
required by \cite{GCL}. Moreover, as \cite{GCL} is primarily
interested in polyphase (nonbinary) sequences, their approach is
not applicable for binary set since it requires $n_1=2$. Our
constructions, on the other hand, can be applied to generate both
binary and nonbinary families.
\end{remark}
\begin{remark}
The construction based on (\ref{eq:direct_prime}) generates
sequences that possess the same correlation properties as those of
the so-called PS sequences \cite{PS}. These sequences are
bound-achieving; they have nonzero AC values only on subperiodic
correlation lags at $\tau=m(T+1)$, $m\in\mathbb{Z}$, and zero CC
for all lags. While the PS sequences require that
gcd$(n_1,n_2)=1$, where $n_1$ is a perfect square, to construct an
$(N,n_2,n_1-1)$ family, our method does not impose any constraint
on $n_1$. Moreover, when $n_1$ is a perfect square, our approach
can generate sequences, which, for the convenience of reference,
are called \textit{PS-like sequences}, that require a
constellation of size lcm$(n_2,\sqrt{n_1})=N /\sqrt{n_1}$ as
opposed to lcm$(n_1,n_2)=N$ required by the PS approach \cite{PS}.
Similarly, we refer to those families derived from
(\ref{eq:direct_prime}) using non-perfect square $n_1$ as
\textit{generalized PS sequences} for these sequences cannot be
generated by the PS method. Some PS-like and generalized PS
sequence sets are given in the following subsection.
\end{remark}

\renewcommand{\arraystretch}{1.5}
\begin{table*}[t]
\caption{Polyphase ZCZ sequence sets with sequence length $N=n_1n_2$ using $n_P$-PSK perfect sequence}
 \centering
 \tabcolsep 0.05in
 \label{tab:perf_invol}
 \footnotesize{
 \begin{tabular}{|C{.8in}|C{.8in}|C{.65in}|C{.65in}|
                  C{.8in}|C{.65in}|C{.65in}|C{.65in}|}
  \hline
  & Torii \cite{Polyphase} & \multicolumn{2}{c|}{Tang \cite{interleaving}} & %Hayashi \cite{perfect_seq_based} &
     Popovic \cite{GCL} & Park \cite{PS} &
   \multicolumn{2}{c|}{\textit{Theorem \ref{poly_by_direct}}} \\ \hline
 Perfect sequence length & $n_1=kt$, $k\leq n_1$ & \multicolumn{2}{c|}{$n_1$} & $N$ & $n_1=k^2$%, a perfect square
 & \multicolumn{2}{c|}{$n_1$} \\ \hline
   Set size $K$ & $k$ & \multicolumn{2}{c|}{$n_2$} &
   $n_2$ & $n_2$ & \multicolumn{2}{c|}{$n_2$} \\ \hline
  ZCZ width $T$ & $(n_1-2)k^{\ell-1}$ & $n_1-1$ & $n_1-2$ &
  $n_1-1$ & $n_1-1$ & $n_1-1$ & $n_1-2$ \\ \hline
  Constraints on $n_2$ & $n_2=k^\ell$, $\ell>1$ & $\gcd(n_1,n_2)$ $=1$ &
  $n_1|n_2$ or $n_2|n_1$ & None & $\gcd(n_1,n_2)$ $=1$ &
  $\gcd(n_1,n_2)$ $=1$ & $\gcd(n_1,n_2)$ $\neq1$ \\ \hline
  Upper-bound (\ref{eq:bound}) achieved? & No & Yes & No & Yes & Yes & Yes & No \\ \hline
  Alphabet size & lcm$(k,n_P)$ & \multicolumn{2}{c|}{lcm$(n_2,n_P)$} & % &
  lcm$(n_1,n_2)$ or lcm$(2n_1,n_2)$ & $N$ & \multicolumn{2}{c|}{lcm$(n_2,n_P)$} \\ \hline
 \end{tabular}}
\end{table*}
\renewcommand{\arraystretch}{1}

\subsection{Examples of Direct Synthesized Sequence Sets}
\label{section:examples}
%\subsubsection{PS-Like Sequences}
\begin{example}(PS-like sequences)~
Following the procedure described in \textit{Theorem
\ref{poly_by_direct}} with $N_r=2$, $N'=9$, $B=(1 0 0 0 0 0 0 0 0
1 0 0 0 0 0 0 0 0)$ and $\mathbf{U}$ being the Sylvester Hadamard matrix
$\mathbf{H}_2$, we obtain $\mathbf{P}={\bf U}\vartriangle
B=[P_0^T, P_1^T]^T$, where
\begin{IEEEeqnarray*}{rCl}
P_0=(1,0,0,0,0,0,0,0,0,&&1,0,0,0,0,0,0,0,0), \nonumber\\
P_1=(1,0,0,0,0,0,0,0,0,&-&1,0,0,0,0,0,0,0,0).~~~~
\end{IEEEeqnarray*}
Filtering them by the upsampled perfect sequence
$A=(W_{3}^{0} 0 W_{3}^{0} 0 W_{3}^{0} 0 W_{3}^{0} 0 W_{3}^{2} 0
W_{3}^{1} 0 W_{3}^{0} 0 W_{3}^{1} 0 W_{3}^{2} 0)$, we have
\begin{IEEEeqnarray}{rCl}
C_{0,0}=P_0 \circ A
  &=&(W_{6}^{0} W_{6}^{2}   W_{6}^{2}   W_{6}^{0}   W_{6}^{4}    W_{6}^{0}  W_{6}^{0}   W_{6}^{0}   W_{6}^{4}%\nonumber\\& &~
  W_{6}^{0}  W_{6}^{2}   W_{6}^{2}   W_{6}^{0}    W_{6}^{4}   W_{6}^{0}  W_{6}^{0}   W_{6}^{0}    W_{6}^{4}), \nonumber\\
C_{0,1}=P_1 \circ A
  &=&(
  W_{6}^{0} W_{6}^{5}   W_{6}^{2}   W_{6}^{3}   W_{6}^{4}    W_{6}^{3}  W_{6}^{0}   W_{6}^{3}   W_{6}^{4}%\nonumber\\& &~
  W_{6}^{3}  W_{6}^{2}   W_{6}^{5}   W_{6}^{0}    W_{6}^{1}   W_{6}^{0}  W_{6}^{3}   W_{6}^{0}    W_{6}^{1}).~~~~
\label{eq:PS-like_1}
\end{IEEEeqnarray}
It can be shown that
\begin{IEEEeqnarray}{rCl}
\theta_{C_{0,0}C_{0,1}}(\tau)=0,~~~
|\theta_{C_{0,0}C_{0,0}}(\tau)|=|\theta_{C_{0,1}C_{0,1}}(\tau)|=18\delta(|\tau|_9).
\nonumber
\end{IEEEeqnarray}
%where $|k|_L\stackrel{def}{=}k$ mod $L$.
and $\mathbf{C}_0=\{C_{0,0},C_{0,1}\}$, is an $(18,2,8)$
bound-achieving ZCZ sequence family.

Using cyclically-shifted basic sequences $B'(n)=B(|n-3|_{18})$ and
$B''(n)=B(|n-6|_{18})$, we obtain two new $(18,2,8)$ ZCZ sequence
sets $\mathbf{C}_1=\{C_{1,0},C_{1,1}\}$ and
$\mathbf{C}_2=\{C_{2,0},C_{2,1}\}$ whose members are
\begin{IEEEeqnarray*}{rCl}
  C_{1,0}&=&(
  W_{6}^{0} W_{6}^{0}   W_{6}^{4}   W_{6}^{0}   W_{6}^{2}    W_{6}^{2}  W_{6}^{0}   W_{6}^{4}   W_{6}^{0}%\nonumber\\& &~
  W_{6}^{0}  W_{6}^{0}   W_{6}^{4}   W_{6}^{0}    W_{6}^{2}   W_{6}^{2}  W_{6}^{0}   W_{6}^{4}   W_{6}^{0}), \nonumber\\
  C_{1,1}&=&(
  W_{6}^{3} W_{6}^{0}   W_{6}^{1}   W_{6}^{0}   W_{6}^{5}    W_{6}^{2}  W_{6}^{3}   W_{6}^{4}   W_{6}^{3}%\nonumber\\& &~
  W_{6}^{0}  W_{6}^{3}   W_{6}^{4}   W_{6}^{3}    W_{6}^{2}   W_{6}^{5}  W_{6}^{0}   W_{6}^{1}   W_{6}^{0}), \nonumber\\
  C_{2,0}&=&(
  W_{6}^{0} W_{6}^{4}   W_{6}^{0}   W_{6}^{0}   W_{6}^{0}    W_{6}^{4}  W_{6}^{0}   W_{6}^{2}   W_{6}^{2}%\nonumber\\& &~
  W_{6}^{0}  W_{6}^{4}   W_{6}^{0}   W_{6}^{0}    W_{6}^{0}   W_{6}^{4}  W_{6}^{0}   W_{6}^{2}   W_{6}^{2}), \nonumber\\
  C_{2,1}&=&(
  W_{6}^{0} W_{6}^{1}   W_{6}^{0}   W_{6}^{3}   W_{6}^{0}    W_{6}^{1}  W_{6}^{0}   W_{6}^{5}   W_{6}^{2}%\nonumber\\& &~
  W_{6}^{3}  W_{6}^{4}   W_{6}^{3}   W_{6}^{0}    W_{6}^{3}   W_{6}^{4}  W_{6}^{3}   W_{6}^{2}   W_{6}^{5}).
\end{IEEEeqnarray*}
It can be shown that $\theta_{BB'}(\tau)=\theta_{BB''}(\tau)=
\theta_{B'B''}(\tau)=0$, $\forall~|\tau|\leq T'=2$ and thus the
inter-set zero-CC zone width is 2. Moreover, the set $\mb{C}\defeq
\bigcup_{i=0}^{2}\mb{C}_i$ is a bound-achieving $(18,6,2)$ ZCZ
sequence set. \label{ex:subs_direct}
\end{example}
\begin{example}(Length-$12$ PS-like sequences)~
The set of three PS-like sequences
%\begin{eqnarray}
%{\mathbf{U}}=\left[%\frac{1}{\sqrt{3}}
%\begin{array}{ccc}
%W_3^0 & W_3^0 & W_3^0 \\
%W_3^0 & W_3^1 & W_3^2 \\
%W_3^0 & W_3^2 & W_3^1
%\end{array}
%\right]
%\end{eqnarray}
\begin{IEEEeqnarray*}{rCl}
P_0&=&(W^0_3 0 0 0 W^0_3 0 0 0 W^0_3 0 0 0), \nonumber \\
P_1&=&(W^0_3 0 0 0 W^1_3 0 0 0 W^2_3 0 0 0), \nonumber\\
P_2&=&(W^0_3 0 0 0 W^2_3 0 0 0 W^1_3 0 0 0)
\label{eq:three_W3_seqs}
\end{IEEEeqnarray*}
is generated by using $N_r=3$, $N'=4$, $B=(1 0 0 0 1 0 0 0 1 0 0 0
)$, and IDFT matrix $\mathbf{U}=\mathbf{F}_3^H$. Filtering them by
$A=(1, 0, 0, 1, 0, 0,1,0, 0,$ $ -1, 0, 0)$, we obtain the ZCZ
sequences
\begin{IEEEeqnarray}{rCl}
C_0&=&P_0 \circ A %\nonumber \\ &=&
=(W_6^0 W_6^0 W_6^0W_6^3 W_6^0
W_6^0 W_6^0  W_6^3 W_6^0  W_6^0  W_6^0  W_6^3), \nonumber\\
C_1&=&P_1 \circ A %\nonumber \\ &=&
=(W_6^0 W_6^2 W_6^4 W_6^3
W_6^2 W_6^4 W_6^0 W_6^5 W_6^4 W_6^0 W_6^2 W_6^1), \nonumber\\
C_2&=&P_2 \circ A %\nonumber\\&=&
=(W_6^0 W_6^4 W_6^2 W_6^3 W_6^4 W_6^2 W_6^0 W_6^1 W_6^2 W_6^0
W_6^4 W_6^5).~~~~
\label{eq:PS-like_2}
\end{IEEEeqnarray}
It is verifiable that $\forall~ i,j$, $i\neq j$,
\begin{IEEEeqnarray}{rCl}
 \theta_{C_iC_j}(\tau)=0,~~~
 |\theta_{C_iC_i}(\tau)|=12 \delta(|\tau|_4),
\end{IEEEeqnarray}
i.e., $\mathbf{C}=\{C_0, C_1, C_2\}$ is a $(12,3,3)$
bound-achieving ZCZ sequence set. This set also possesses the same
PS sequence correlation properties \cite{PS}. Moreover, both
(\ref{eq:PS-like_1}) and (\ref{eq:PS-like_2}) require only $1/3$
and $1/2$ of the alphabet size required by the original PS
construction under the same sequence period constraint.
\end{example}

\begin{example}(Generalized PS sequences)~
Using the method of \textit{Theorem \ref{poly_by_direct}} with
$N_r=5$, $N'=3$, the IDFT matrix $\mathbf{U}=\mathbf{F}_5^H$,
$B=(1 0 0 1 0 0 1 0 0 1 0 0 1 0 0)$, and $A=(W_3^0 0000 W_3^2 0000
W_3^0 0000)$, we obtain
\begin{IEEEeqnarray*}{rCl}
C_0%&=&P_0^0 \circ A \nonumber \\
&=&(W_{15}^0  W_{15}^5  W_{15}^0
W_{15}^0  W_{15}^5  %\nonumber \\ &&~
W_{15}^0  W_{15}^0 W_{15}^5
W_{15}^0 W_{15}^0 W_{15}^5
W_{15}^0  W_{15}^0  W_{15}^5  W_{15}^0), \nonumber\\
C_1%&=&P_1^0 \circ A \nonumber \\
&=&(W_{15}^0  W_{15}^{11}
W_{15}^{12} W_{15}^3  W_{15}^{14} %\nonumber \\ &&~
W_{15}^0
W_{15}^6 W_{15}^2 W_{15}^3 W_{15}^9 W_{15}^5
W_{15}^6  W_{15}^{12} W_{15}^8  W_{15}^9), \nonumber\\
C_2%&=&P_2^0 \circ A \nonumber\\
&=&(W_{15}^0  W_{15}^2  W_{15}^9  W_{15}^6  W_{15}^8  %\nonumber \\&&~
W_{15}^0 W_{15}^{12}
W_{15}^{14} W_{15}^6  W_{15}^3  W_{15}^5  W_{15}^{12} W_{15}^9  W_{15}^{11} W_{15}^3),\nonumber\\
C_3%&=&P_3^0 \circ A \nonumber\\
&=&(W_{15}^0  W_{15}^8  W_{15}^6  W_{15}^9  W_{15}^2 %\nonumber \\&&~
W_{15}^0 W_{15}^3
W_{15}^{11} W_{15}^9  W_{15}^{12} W_{15}^5  W_{15}^3  W_{15}^6  W_{15}^{14} W_{15}^{12}),\nonumber\\
C_4%&=&P_4^0 \circ A \nonumber\\
&=&(W_{15}^0  W_{15}^{14} W_{15}^3  W_{15}^{12} W_{15}^{11} %\nonumber \\&&~
W_{15}^0 W_{15}^9 W_{15}^8 W_{15}^{12} W_{15}^6  W_{15}^5
W_{15}^9  W_{15}^3  W_{15}^2  W_{15}^6)~
\end{IEEEeqnarray*}
which constitute a set of $(15,5,2)$ bound-achieving generalized
PS sequences that has the same correlation properties as the
original PS sequences, i.e., $\forall~ i,j$, $i\neq j$,
\begin{IEEEeqnarray}{rCl}
\theta_{C_iC_j}(\tau)~&=0,~~~ |\theta_{C_iC_i}(\tau)|&=15
\delta(|\tau|_3).
\end{IEEEeqnarray}
As mentioned before, the PS method \cite{PS} can not produce ZCZ
sequences of length $N=15$.
\end{example}
%\subsubsection{Length-$24$ Sequence Sets}

Previous examples are constructed by using coprime $N_r$ and $N'$,
we show a set using the construction (\ref{eq:direct_notprime}).
\begin{example}(Sets based non-coprime parameters and nested-like
sets using orthogonal tone decomposition)
~ By choosing $N_r=4$, $N'=6$ and upsampling the
Sylvester Hadamard $\mathbf{H}_4$ by $B=(1 0 0 0 0 0 1 0 0 0 0 0 0
1 0 0 0 0 0 1 0 0 0 0)$, we obtain a $(24,4,4)$ ZCZ sequence family
by filtering each row of $\mathbf{P}=\mathbf{H}_4 \vartriangle B$
through $A=(W_{12}^{0}000W_{12}^{1}000W_{12}^{4}000 W_{12}^{9} 0 0
0 W_{12}^{4} 0 0 0 W_{12}^{1} 0 0 0)$:
\begin{IEEEeqnarray*}{rCl}
C_0=P_0 \circ A %\nonumber \\
&=&\hspace{-.2em} (W_{12}^{0}  W_{12}^{3}  W_{12}^{11} W_{12}^{8}
W_{12}^{11} W_{12}^{8}  W_{12}^{0}  W_{12}^{3}  W_{12}^{8}
W_{12}^{11}
W_{12}^{11} W_{12}^{8}\nonumber \\
&&~\hspace{-.2em} W_{12}^{3}  W_{12}^{0} W_{12}^{8} W_{12}^{11}
W_{12}^{8} W_{12}^{11} W_{12}^{3} W_{12}^{0} W_{12}^{11}
W_{12}^{8} W_{12}^{8}  W_{12}^{11}), \nonumber\\
C_1=P_1 \circ A %\nonumber \\
&=&\hspace{-.2em} (W_{12}^{0}  W_{12}^{3}  W_{12}^{5}  W_{12}^{2}
W_{12}^{11} W_{12}^{8}  W_{12}^{6}  W_{12}^{9}  W_{12}^{8}
W_{12}^{11} W_{12}^{5}  W_{12}^{2}\nonumber \\
&&~\hspace{-.2em} W_{12}^{3}  W_{12}^{0} W_{12}^{2}  W_{12}^{5}
W_{12}^{8}  W_{12}^{11} W_{12}^{9} W_{12}^{6}  W_{12}^{11}
W_{12}^{8}  W_{12}^{2}  W_{12}^{5}), \nonumber\\
C_2=P_2 \circ A %\nonumber \\
&=&\hspace{-.2em} (W_{12}^{0}  W_{12}^{9}  W_{12}^{11} W_{12}^{2}
W_{12}^{11} W_{12}^{2}  W_{12}^{0}  W_{12}^{9}  W_{12}^{8}
W_{12}^{5}  W_{12}^{11} W_{12}^{2}\nonumber \\
&&~\hspace{-.2em} W_{12}^{3}  W_{12}^{6} W_{12}^{8}  W_{12}^{5}
W_{12}^{8}  W_{12}^{5}  W_{12}^{3} W_{12}^{6}  W_{12}^{11}
W_{12}^{2}  W_{12}^{8}  W_{12}^{5}), \nonumber\\
C_3=P_3 \circ A %\nonumber \\
&=&\hspace{-.2em} (W_{12}^{0}  W_{12}^{9}  W_{12}^{5}  W_{12}^{8}
W_{12}^{11} W_{12}^{2}  W_{12}^{6}  W_{12}^{3}  W_{12}^{8}
W_{12}^{5}  W_{12}^{5}  W_{12}^{8}\nonumber \\
&&~\hspace{-.2em} W_{12}^{3}  W_{12}^{6} W_{12}^{2}  W_{12}^{11}
W_{12}^{8}  W_{12}^{5}  W_{12}^{9} W_{12}^{0}  W_{12}^{11}
W_{12}^{2}  W_{12}^{2}  W_{12}^{11}).
\end{IEEEeqnarray*}

Alternatively, we can perform orthogonal tone decomposition on $B$
to obtain two weight-$2$ basic sequences of same minimum spacing
$10$:
\begin{IEEEeqnarray*}{rCl}
B_0&=&(100000000000010000000000),\\
B_1&=&(000000100000000000010000).
\end{IEEEeqnarray*}
With $\mb{U}_0=\mb{H}_2$ and
\begin{IEEEeqnarray*}{rCl}
\mb{U}_1&=&\left[
\begin{array}{cc}
1 & j\\
j & 1\\
\end{array}
\right],
\end{IEEEeqnarray*}
we filter rows of
$\mathbf{U}_0\vartriangle B_0$ and
$\mathbf{U}_1\vartriangle B_1$ by \cite{poly_perf_prop_2}
\begin{IEEEeqnarray*}{rCl}
A&=&(W_{6}^{0} 0 W_{6}^{0} 0 W_{6}^{3}0 W_{6}^{2}0
W_{6}^{4}0 W_{6}^{2}0  %\\&&~
W_{6}^{3}0  W_{6}^{6}0
W_{6}^{6}0  W_{6}^{2}0 W_{6}^{1}0 W_{6}^{2}0)
\end{IEEEeqnarray*}
to obtain two smaller polyphase sets of larger ZCZ width,
$\mb{C}_0=\{C_{0,0},C_{0,1}\}$  and $\mb{C}_1=\{C_{1,0},C_{1,1}\}$,
where
\begin{IEEEeqnarray*}{rCl}
C_{0,0}&=& (
W_{6}^{0}  W_{6}^{3}  W_{6}^{4} W_{6}^{4}
W_{6}^{5} W_{6}^{2}  W_{6}^{4}  W_{6}^{4}  %
W_{6}^{0} W_{6}^{3} W_{6}^{0} W_{6}^{0}\\&&~
W_{6}^{3}  W_{6}^{0} W_{6}^{4} W_{6}^{4}%\\&&~
W_{6}^{2} W_{6}^{5} W_{6}^{4} W_{6}^{4} W_{6}^{3}
W_{6}^{0} W_{6}^{0}  W_{6}^{0}),\\
C_{0,1}&=& (W_{6}^{0}  W_{6}^{0}  W_{6}^{4} W_{6}^{1}
W_{6}^{5} W_{6}^{5}  W_{6}^{4}  W_{6}^{1}  %
W_{6}^{0} W_{6}^{0} W_{6}^{0} W_{6}^{3}\\&&~
W_{6}^{3}  W_{6}^{3} W_{6}^{4} W_{6}^{1} %\\&&~
W_{6}^{2} W_{6}^{2} W_{6}^{4} W_{6}^{1} W_{6}^{3}
W_{6}^{3} W_{6}^{0}  W_{6}^{3}),\\
C_{1,0}&=& (W_{12}^{8}  W_{12}^{11}  W_{12}^{6} W_{12}^{3}
W_{12}^{0} W_{12}^{3}  W_{12}^{0}  W_{12}^{9}%
W_{12}^{8}
W_{12}^{11} W_{12}^{10} W_{12}^{7}\\&&~   W_{12}^{8}  W_{12}^{11} W_{12}^{0} W_{12}^{9} %\\&&~
W_{12}^{0} W_{12}^{3} W_{12}^{6} W_{12}^{3} W_{12}^{8}
W_{12}^{11} W_{12}^{4}  W_{12}^{1}),\\
C_{1,1}&=& (W_{12}^{11}  W_{12}^{8}  W_{12}^{9} W_{12}^{0}
W_{12}^{3} W_{12}^{0}  W_{12}^{3}  W_{12}^{6}  %
W_{12}^{11} W_{12}^{8} W_{12}^{1} W_{12}^{4} \\&&~
W_{12}^{11}  W_{12}^{8} W_{12}^{3} W_{12}^{6} %\\&&~
W_{12}^{3} W_{12}^{0} W_{12}^{9} W_{12}^{0} W_{12}^{11}
W_{12}^{8} W_{12}^{7}  W_{12}^{10}).
\end{IEEEeqnarray*}
Both sets are $(24,2,10)$ ZCZ sequence sets and together they form
another $(24,4,4)$ set.
\end{example}
\begin{example}(Generalized PS sequence set)~ We can also derive a
smaller {\it generalized PS sequence} set having the same period
but a larger ZCZ width. For example, if we choose $N'=8$, $N_r=3$,
(\ref{eq:direct_prime}), and use the QPSK perfect sequence \cite{Mow_diss}
\begin{IEEEeqnarray*}{rCl}
A'=(W_4^1W_4^1W_4^2W_4^1W_4^1W_4^3W_4^2W_4^3),
\end{IEEEeqnarray*}
then the three sequences
\begin{IEEEeqnarray*}{rCl}
C_0&=&(W_{12}^{9}  W_{12}^{3}  W_{12}^{6}  W_{12}^{3}  W_{12}^{9}
W_{12}^{9}  W_{12}^{6}  W_{12}^{9}  W_{12}^{9}  W_{12}^{3}
W_{12}^{6}  W_{12}^{3} \nonumber \\
&&~W_{12}^{9}  W_{12}^{9} W_{12}^{6} W_{12}^{9}  W_{12}^{9}
W_{12}^{3}  W_{12}^{6}
W_{12}^{3} W_{12}^{9}  W_{12}^{9}  W_{12}^{6}  W_{12}^{9} ),\nonumber \\
C_1&=&(W_{12}^{9}  W_{12}^{7}  W_{12}^{2}  W_{12}^{3}  W_{12}^{1}
W_{12}^{5}  W_{12}^{6}  W_{12}^{1}  W_{12}^{5}  W_{12}^{3}
W_{12}^{10} W_{12}^{11} \nonumber \\
&&~W_{12}^{9}  W_{12}^{1} W_{12}^{2} W_{12}^{9}  W_{12}^{1}
W_{12}^{11} W_{12}^{6}
W_{12}^{7} W_{12}^{5}  W_{12}^{9}  W_{12}^{10} W_{12}^{5}), \nonumber \\
C_2&=&(W_{12}^{9}  W_{12}^{11} W_{12}^{10} W_{12}^{3}  W_{12}^{5}
W_{12}^{1}  W_{12}^{6}  W_{12}^{5}  W_{12}^{1}  W_{12}^{3}
W_{12}^{2}  W_{12}^{7} \nonumber \\
&&~W_{12}^{9}  W_{12}^{5} W_{12}^{10} W_{12}^{9}  W_{12}^{5}
W_{12}^{7}  W_{12}^{6} W_{12}^{11} W_{12}^{1}  W_{12}^{9}
W_{12}^{2} W_{12}^{1})
\end{IEEEeqnarray*}
constitute an $(N,N_r,N'-1)=(24,3,7)$ bound-achieving ZCZ family.
A family with such ZCZ parameter values can not be generated by
the method suggested in \cite{PS}.
\end{example}

%\subsubsection{Binary and Ternary ZCZ Sequences}
%While ZCZ sequences presented in the previous examples begin with
%a basic sequence that has equally-spaced nonzero entries, one can
%also build ZCZ sequences based on nonuniformly-spaced basic
%sequences.
\begin{example}(Length-$16$ ternary and binary sequences)~
Using the basic sequence $B=(1 0 0 0 0$ $0 0 1 0 0 1 0 0 1 0 0)$,
the Sylvester Hadamard matrix $\mathbf{H}_4$ as $\mathbf{U}$, and the perfect
sequence $A=(+\:0 0 0 + 0 0 0 + 0 0 0 - 0 0 0)$, we obtain
\begin{IEEEeqnarray*}{rCl}
P_0&=&(+\: 0 0 0 0 0 0 + 0 0 + 0 0 + 0 0), \nonumber \\
P_1&=&(+\: 0 0 0 0 0 0 - 0 0 + 0 0 - 0 0), \nonumber \\
P_2&=&(+\: 0 0 0 0 0 0 + 0 0 - 0 0 - 0 0), \nonumber \\
P_3&=&(+\: 0 0 0 0 0 0 - 0 0 - 0 0 + 0 0),
\end{IEEEeqnarray*}
where $+$ and $-$ denote $+1$ and $-1$, respectively. Time domain
sequences with zero entries are often undesirable as they require
on-off switching. Filtering $\{P_i\}$ by $A$, we obtain the binary
$(16,4,2)$ ZCZ sequence family consisting of
\begin{IEEEeqnarray*}{rCl}
C_0&=&P_0 \circ A = (+ - + + - + + +
 + + + - + + - \:+\hspace{.1em}),\nonumber\\
C_1&=&P_1\circ A =(+ + + - - - +
- + - + + + - - \:-\hspace{.1em}),\nonumber\\
C_2&=&P_2\circ A=(+ + - + - - -
+ + - - - + - + \:+\hspace{.1em}),\nonumber\\
C_3&=&P_3\circ A=(+ - - - - + - -
+ + - + + + + \:-\hspace{.1em}).
\end{IEEEeqnarray*}
\label{ex:BiDirect1}
\end{example}
\begin{example}(Length-$32$ binary sequence set)~
Let $N_r=8$ and $N'=4$. With $B=(1 0 0 0$ $ 1 0 0 0 0 0 0 1 0 0 01 0 0 1
0 0 0 1 0 0 1 0 0 0 1 0 0)$, $A=(+\: 0 0 0 0 0 0 0 + 0 0 0 0 0
0 0  +  0 0 0 0 0 0 0 - 0 0 0 0 0 0 0 )$ and $\mb{U}=\mathbf{H}_8$, we obtain
the binary $(32,8,2)$ ZCZ sequence set
\begin{IEEEeqnarray*}{rCl}
C_0=(\hspace{-.2em}&+&   -  +   +   +   -  +   +   -  +   +   + -
+ + +%\nonumber\\&+&
++   +   -  +   +   +   -  +   +   -  +   + + -
+\hspace{.1em}
),\nonumber\\
C_1=(\hspace{-.2em}&+&   -  +   +   -  +   -  -  -  +   +   +   +
-  - -%\nonumber\\ &+&
++ + -  -  -  -  +   +   +   -  +   -  -  +
-\hspace{.1em}
),\nonumber\\
C_2=(\hspace{-.2em}&+&   +   +   -  +   +   +   -  -  -  +   -  -
-  + -%\nonumber\\ &+&
+ - + +   +   -  +   +   +   -  -  -  +   -  -
-\hspace{.1em}
),\nonumber\\
C_3=(\hspace{-.2em}&+&   +   +   -  -  -  -  +   -  -  +   -  +
+ - +%\nonumber\\ &+&
+ - + + -  +   -  -  +   -  -  -  -  +   +
+\hspace{.1em}
),\nonumber\\
C_4=(\hspace{-.2em}&+&   +   -  +   +   +   -  +   -  -  -  +   -
- - +%\nonumber\\ &+&
+ - - - +   -  -  -  +   -  +   +   +   -  + +
\hspace{.1em}
),\nonumber\\
C_5=(\hspace{-.2em}&+&   +   -  +   -  -  +   -  -  -  -  +   + +
+ -%\nonumber\\ &+&
+ - - - -  +   +   +   +   -  +   +   -  + -  -
\hspace{.1em}),\nonumber\\
C_6=(\hspace{-.2em}&+&   -  -  -  +   -  -  -  -  +   -  -  -  + -
- %\nonumber\\ &+&
+ + - + + +   -  +   +   +   +   -  +   +   + -
\hspace{.1em}),\nonumber\\
C_7=(\hspace{-.2em}&+&   -  -  -  -  +   +   +   -  +   -  -  + -
+ + %\nonumber\\ &+&
+ + - + -  -  +   -  +   +   +   -  -  -  - +
\hspace{.1em}).
\end{IEEEeqnarray*}
\label{ex:BiDirect2}
\end{example}

The ZCZ families shown in the above two examples achieve
(\ref{eq:bi_bound}), the bound for binary ($N_{A'}=2$) sequences,
but their ZCZ widths are limited by the facts that there exists
only one binary perfect sequence (whose length $N'=4$) and binary
Hadamard matrices only exists for certain $N_r$; see {\it Remark
\ref{prop:Hada_mtx_L}}. To increase the ZCZ width and have greater
flexibility in choosing the ZCZ parameters, we can use
higher-order constellations ($N_r>2$). For example, quadriphase
perfect sequences of length $N'=2$, $4$, $8$ or $16$ do exist
\cite{poly_perf_prop_2}, \cite{Mow_diss}. We introduce in the next
section an alternate method which offers more choices for the ZCZ
width.

%
%The new polyphase ZCZ sequences generated from
%(\ref{eq:direct_notprime}) are {\it generalizations} of
%sequences reported in \cite{Polyphase} where $N_r$ is
%constrained to be a multiple or a factor of $N'$. In contrast, the
%proposed method is more flexible in that we can produce sequences
%of {\it any length that is equal to a composite number $N=N_r
%N'$}, where $N'$ is the length of the polyphase perfect sequence
%and $N_r$ is any integer.

%Note that given a unitary matrix, a perfect sequence
%$A'$ and the length of ZCZ sequences, the method proposed
%in \cite{Polyphase} can generate one ZCZ family only
%while our approach can generate many ZCZ families due to the use
%of cyclically shifted versions of the original basic sequence as
%basic sequences.

%%%Moreover, some of the sequences generated by the method suggested
%%%in \cite{newpoly} are similar to those derived from
%%%(\ref{eq:direct_prime}) with the constraint that $N_r$
%%%divides $N'+1$. As a result, one can generate same polyphase ZCZ
%%%sequences as those given in \cite{newpoly} by employing
%%%\textit{Theorem \ref{poly_by_direct}} and a properly
%%%rearranged basic sequence $B$. % and thus produce sequences.
%%%Note that such sequences have zero-correlation zone widths smaller
%%%than those generated by (\ref{eq:direct_notprime}) and
%%%(\ref{eq:direct_prime}).

\section{Sequences Derived from Complementary Sets of Sequences}
\label{section:comp} In this section, we generalize the above
basic sequence based approach by replacing rows of an unitary
matrix with concatenated sequences. The following definitions can
be found in \cite{complementary}.
\subsection{Basic Definitions}
\begin{definition}
The aperiodic CC function of two length-$L$ sequences $u \equiv
\{u(n)\}$ and $v \equiv \{v(n)\}$ is defined as
\begin{eqnarray}
\psi_{uv}(\tau)= \sum_{n=\tau}^{L-1}u(n)v^*(n-\tau).
\end{eqnarray}
The aperiodic AC function of sequence $u$ is obviously
$\psi_{uu}(\tau)$.
\end{definition}
%It is worth mentioning that these aperiodic correlation functions
%defined here are to aid the introduction of the subsequent
%definitions and then elicit the ZCZ sequence construction method
%proposed in the next subsection. \textit{Sequences constructed by
%such method are still periodic}.
%是為了介紹complementary set  真正的sequence依然還是periodic
\begin{definition}
A set of $Q$ equal-length sequences,
$\mathbf{E}=\{E_0,E_1,\cdots,E_{Q-1}\}$, forms a
\textit{complementary set} of sequences (CSS) if and only if
$\forall\tau \neq 0$,
\begin{eqnarray}
 \sum_{i=0}^{Q-1}\psi_{E_{i}E_{i}}(\tau)=0.
\end{eqnarray}
%For the special case of binary sequences, a set is said to be {\it
%complementary} if the total number of pairs of like elements with
%a given separation is equal to the total number of pairs of unlike
%elements with the same separation in these sequences.
\label{def:comp_set}
\end{definition}
\begin{definition}
A CSS, $\mathbf{F}=\{F_0,F_1,$ $\cdots,F_{Q-1}\}$, is said to be a
\textit{mate} of the CSS, $\mathbf{E}=$
$\{E_0,E_1,\cdots,E_{Q-1}\}$ if
\begin{enumerate}
 \item[(a)]{The lengths of all members in $\mathbf{E}$ and
  $\mathbf{F}$ are the same;}
 \item[(b)]{For all $\tau$,
\begin{equation}
  \sum_{i=0}^{Q-1} \psi_{E_i F_i}(\tau)=0.
\end{equation}}
\end{enumerate}
\label{def:mate}
\end{definition}
\begin{definition}
A collection of complementary sets of sequences $\{\mathbf{E}_0,
\mathbf{E}_1,\cdots,\mathbf{E}_{K-1}\}$, where each set contains
the same number of sequences, is said to be \textit{mutually
orthogonal} if every two sets in the collection are mates of each
other. \label{def:mut_ortho}
\end{definition}

It has been proved in \cite{generalcom} that
\begin{corollary}
The number of mutually orthogonal CSS's (MOCSS's) $K$ cannot
exceed the cardinality of member CSS, $Q$, i.e., $K\leq Q$.
\end{corollary}

\subsection{Synthesis Procedure}
We now extend the nonuniform upsampling operation defined in
\textit{Definition \ref{def:nonuni_upsamp}}.
\begin{definition}
Let $V$ be a length-$N$ binary sequence with $w_H(V)=Q$ and
$\mathcal{E}=\{\mathbf{E}_0,\mathbf{E}_1,\cdots,\mathbf{E}_{K-1}\}$
be a collection of $K$ MOCSS's in which each CSS ${\mathbf E}_i$
consists of $Q$ length-$L$ sequences, i.e., ${\mathbf
E}_i=\{E_{i,0}, E_{i,1},\cdots,E_{i,Q-1}\}$, where
$E_{i,j}=\left(e_{i,j}(0),e_{i,j}(1),\cdots,e_{i,j}(L-1)\right)$.

The $V$-upsampled concatenated sequence based on $\mathbf{E}_i$,
$G_i=\mathbf{E}_i \vartriangle_c V=
(g_i(0),g_i(1),\cdots,g_i(N+Q(L-1)-1))$ is defined by
%\begin{eqnarray}
% g_{i,j(\rho-1)+s_v(j)+m}= \left\{
% \begin{array}{ll}
% e_{i,j}(m), & s_v(j)=\mbox{the coordinate of}\\
%   \     &  \ \ \ ~~~~~~~\mbox{the $j$th nonzero}\\
%   \     &  \ \ \ ~~~~~~~\mbox{entry of $V$};\\
% 0, & \mbox{otherwise}. \\
% \end{array}
% \right.
% \nonumber
% \end{eqnarray}
 \begin{eqnarray}
 g_{i}(n)= \left\{
 \begin{array}{ll}
 e_{i,j}(m), & n=j(L-1)+s_V(j)+m,\\
%   \     &  \mbox{};\\
 0, & \mbox{otherwise}, \\
 \end{array}
 \right.
 \end{eqnarray}
where $s_V(j)$ is given in {\it Definition \ref{def:nonuni_upsamp}}.
\label{def:nonuni_upsamp_comp}
\end{definition}
The operator $\vartriangle_c$ is similar to $\vartriangle$: the
latter operates on rows of a matrix while the former operates on
the sequence formed by concatenating members of the set
$\mathbf{E}_i$ and replaces each nonzero element of a basic
sequence by a finite-length sequence.
\begin{lemma}
Let $\mathcal{E}=\{\mathbf{E}_0,\mathbf{E}_1,\cdots,$
$\mathbf{E}_{K-1}\}$ be a collection of $K$ MOCSS's in which each
set ${\mathbf E}_i$ has $Q$ length-$L$ sequences and $B$ be a
basic $(N,T)$ sequence of weight $Q$. The set $\mathbf{G}=\{{\bf
E}_i \vartriangle_c B\}\defeq\{G_0, G_1, \cdots,G_{K-1}\}$ forms
an $(N+Q(L-1), K,T)$ ZCZ sequence family. \label{comp}
\end{lemma}
\begin{proof}
Based on \textit{Lemma \ref{basic_seq}} and \textit{Definition
\ref{def:nonuni_upsamp_comp}} we can express $G_i$ as
\begin{IEEEeqnarray*}{rCl}
 G_i=(\underbrace{0 \cdots 0}_{s_V(0)}~
      &&\underbrace{e_{i,0}(0) \cdots  e_{i,0}(L-1)}_L \underbrace{0 \cdots 0 }_{s_V(1)-s_V(0)-1}\nonumber\\
      &&\underbrace{e_{i,1}(0) \cdots  e_{i,1}(L-1)}_L \underbrace{0 \cdots 0 }_{s_V(2)-s_V(1)-1}\nonumber\vspace{-.8em}\\
      &&~~~~~~\vdots\nonumber\vspace{-.5em}\\%0 \cdots 0~~~~~~~~~~~~~~~
      &&\underbrace{e_{i,Q-1}(0) \cdots e_{i,Q-1}(L-1)}_L
      \overbrace{0 \cdots 0}^{N-s_V(Q-1)-1}\hspace{-1em}),\nonumber
\end{IEEEeqnarray*}
where $s_V(j)-s_V(j-1)-1\geq T$, $j=0,1,\cdots,Q-1$, and
$s_V(0)-s_V(Q-1)+N-1\geq T$. Invoking \textit{Definitions
\ref{def:mate}} and \textit{\ref{def:mut_ortho}}, we obtain,
for all $i\neq k$,
\begin{IEEEeqnarray}{rCl}
 \theta_{G_iG_k}(\tau)
  =\sum_{j=0}^{Q-1}
  \psi_{E_{i,j}E_{k,j}}(\tau)=0,~|\tau|_N\leq T.
\end{IEEEeqnarray}
By analogy, \textit{Definition \ref{def:comp_set}} gives, for all
$i$, $\theta_{G_iG_i}(\tau)=\sum_{j=0}^{Q-1}
\psi_{E_{i,j}E_{i,j}}(\tau)=0$, $0 < |\tau|_N\leq T$. Therefore,
$\mathbf{G}$ forms an $(N+Q(L-1),K,T)$ ZCZ sequence set.
%$\mathbf{G}$ is also an $(N+Q(L-1),$ $K,T)$ ZCZ sequence set as
%\textit{Lemma \ref{mod_inv}} ensures that the correlation
%properties after filtering all $G_i$'s by $A$ remain unchanged.
\end{proof}

\subsection{Polyphase ZCZ Sequences}
Following the idea described in Section \ref{section:direct_poly},
we can derive another class of polyphase ZCZ sequence families by
using suitable perfect and basic sequences. The proof of the next
corollary is similar to that of {\it Theorem \ref{poly_by_direct}}
and is given in the last two paragraphs of \ref{app:pf_direct}.
\begin{corollary}
Let $A$ be the length-$L N$ perfect sequence obtained by $L
N_r$-fold upsampling on a length-$N'$ perfect $N_{A'}$-PSK
sequence, $A'$, where $N=N_rN'$ and $2 \leq N_{A'}\leq N'$. Denote
by $\mathcal{E}=\{\mathbf{E}_0,\mathbf{E}_1, \cdots,
\mathbf{E}_{K-1}\}$ a collection of $K$ MOCSS's, where $K\leq N_r$
and each CSS ${\mathbf
E}_i=\{E_{i,0},E_{i,1},\cdots,E_{i,N_r-1}\}$ contains $N_r$
length-$L$ $N_c$-PSK sequences. An $(L N,K,T)$ ZCZ $M$-PSK
sequence set, $M=\text{lcm}(N_{A'},N_c)$, with $T=L(N'-2)$ if
$\text{gcd}(N_r,N')\neq1$ or $T=L(N'-1)$ if $\text{gcd}(N_r,N')=1$
can be obtained by the following steps:
\begin{enumerate}
\item[1)] Generate $K$ length-$(N+N_r(L-1))$ sequences
$G'_i=\mathbf{E}_i \vartriangle_c B$, $i=0,1,\cdots, K-1$, where
$B$ is the weight-$N_r$ basic sequence of length $N$ defined by
(\ref{eq:direct_notprime}) if $\text{gcd}(N_r,N')\neq1$ or by
(\ref{eq:direct_prime}) if $\text{gcd}(N_r,N')=1$.

\item[2)]Replace each zero in ${G}'_i$ by a length-$L$ all-zero
sequence to obtain the augmented sequence $G_i$.

\item[3)]{Filter each $G_i$ by $A$.}
\end{enumerate}

%\begin{enumerate}
%\item ${\bf C}$ is a $(\rho N,K,\rho(N'-2))$ polyphase ZCZ family, if $\frac{N}{L_0}>1$.
%\item ${\bf C}$ is a $(\rho N,K,\rho(N'-1))$ polyphase ZCZ
%family if $\frac{N}{L_0}=1$.
%\end{enumerate}
\label{poly_by_comp}
\end{corollary}

%Similar to the property mentioned in {\it Remark 5}, the basic
%sequence $B$ can be cyclically shifted to generate different polyphase
%ZCZ families without changing the correlation properties.

%Note that this construction is also more \textit{flexible} than
%that of \cite{Polyphase} in that the restriction on the
%sequence period is looser.

We have the following four remarks on the MOCSS-based approach.
\begin{remark} Similar to {\it Remark \ref{prop:shif_new}}, the
basic sequence $B$ can be cyclically shifted to generate different
polyphase ZCZ families with the same ZCZ parameters and alphabet
size. These families can be combined to form a larger family with
smaller ZCZ width. Likewise, the zero-CC zone width between $B$
and its shifted version $B'$ determine the inter-set zero-CC zone
width between the associated families or the ZCZ width of the
combined set. We can also decompose $B$ into several basic
sequences $\{B_i\}$ to generate multiple sets with different ZCZ
widths.
\end{remark}
\begin{remark}
As mentioned in Section \ref{section:direct}, binary sequence sets
constructed by \textit{Theorem \ref{poly_by_direct}} have less
choices in ZCZ width. The construction described in {\it Corollary
\ref{poly_by_comp}} takes advantage of the fact that the member
sequence of an MOCSS exists for many values of $L$ and thus allow
the ZCZ width to be chosen from the set $\{T=2L\}$ with the same
basic sequence $B$, set cardinality $N_r$, and perfect sequence
$A'$.
\end{remark}
\begin{remark}
\cite{comp_ex_1} and \cite{comp_ex_2o} present MOCSS-based methods
for generating binary ZCZ sequences. The approach given in
\cite{comp_ex_2o} was later generalized by \cite{comp_ex_4}. The
ZCZ parameters realizable by these methods can be obtained by
using our approach described above. For example, a method given in
\cite{comp_ex_1} needs to use a class of recursively generated
families of binary CSS $\{\mathbf{\Delta}_n\}$. Expressing a
family of $Q$ MOCSS's in matrix form \cite{complementary}
\begin{IEEEeqnarray}{rCl}\mathbf{\Delta}_1
   &\defeq&\left[
   \begin{array}{cccc}
    E_{0,0} & E_{1,0} & \cdots & E_{Q-1,0}\\
    E_{0,1} & E_{1,1} & \cdots & E_{Q-1,1}\\
    \vdots & \vdots & \ddots & \vdots \\
    E_{0,Q-1} & E_{1,Q-1} & \cdots & E_{Q-1,Q-1}\\
   \end{array}\right]
  \end{IEEEeqnarray}
where $E_{i,j}$ are length-$L$ binary sequences and each row is a
CSS. Then, for $n\geq 2$,
  \begin{IEEEeqnarray}{rCl}
   \mathbf{\Delta}_n&=&\left[
   \begin{array}{cc}
    \mathbf{\Delta}_{n-1}\diamond\mathbf{\Delta}_{n-1} & -\mathbf{\Delta}_{n-1}\diamond\mathbf{\Delta}_{n-1}\\
    -\mathbf{\Delta}_{n-1}\diamond\mathbf{\Delta}_{n-1} & \mathbf{\Delta}_{n-1}\diamond\mathbf{\Delta}_{n-1}
   \end{array}\right],
  \end{IEEEeqnarray}
where $[\mathbf{A}\diamond\mathbf{B}]_{ij}$, the $(i,j)$th entry
of the submatrix $[\mathbf{A}\diamond\mathbf{B}]$, is obtained by
concatenating the two sequences, $[\mathbf{A}]_{ij}$ and
$[\mathbf{B}]_{ij}$. The concatenation of rows of
$\mathbf{\Delta}_n$ forms a $(4^{n-1}LQ,2^{n-1}Q,2^{n-2}L)$ ZCZ
sequence set. On the other hand, by using $A'=(1,1,1,-1)$, the
basic sequence defined by (\ref{eq:direct_notprime}) and the
family of MOCSS $\mathbf{\Delta}_{n}$ with $N_r=2^{n-1}Q$ and
elements of $\mathbf{\Delta}_1$ being length $L/4$ sequences, we
obtain binary ZCZ sequence sets with the parameters
$(4^{n-1}4(L/4)Q,2^{n-1}Q,2^{n-1}(L/4)\cdot2)=(4^{n-1}LQ,2^{n-1}Q,2^{n-2}L)$
via \textit{Corollary \ref{poly_by_comp}}.
\end{remark}
\begin{remark}
Our approach offers more choices in parameter values and thus
produce sets which are not derivable from the methods of
\cite{comp_ex_1,comp_ex_4}. More importantly, we can generate not
only binary but also nonbinary sequences and the ZCZ parameters
for the nonbinary class can be flexibly controlled via $N'$, which
can be any integer and is not affected by the MOCSS chosen.
\end{remark}
In Table \ref{tab:MOCSS_bi} we list key parameters for our and
some other MOCSS-based binary ZCZ sequence set constructions.

\renewcommand{\arraystretch}{1.5}
\begin{table}[t]
 \caption{ZCZ sequence sets
 (each uses a collection of $M$ MOCSS's of $Q$ length-$L$ sequences}
 \centering
 \tabcolsep 0.03in
 \label{tab:MOCSS_bi}
 \footnotesize{
 \begin{tabular}{|C{1in}|C{1in}|C{1in}|C{1in}|
                  C{1in}|} \hline
& \multicolumn{3}{c|}{Deng \cite{comp_ex_1} and Tang
\cite{comp_ex_4}, $n\geq0$} & \textit{Corollary
\ref{poly_by_comp}} \\ \hline
 Sequence length $N$ & $4^nLQ$ & $2^{2n-1}LQ$ & $2LQ$ & $4LQ$ \\ \hline
 Set size $K$ & $2^nM$ & $2^nM$ & $M$ & $M$ \\ \hline
 ZCZ width $T$ & $2^{n-1}L$ & $2^{n-2}L$ & $L$ & $2L$ \\ \hline
% Constraints on $n_2$ & ~ &
%  $\gcd(n_1,n_2)=1$ & $n_1|n_2$ or $n_2|n_1$ &
%  $\gcd(n_1,n_2)$ $=1$  \\ \hline
 Upper-bound (\ref{eq:bi_bound}) achieved with $M=Q$? &
  Yes & Yes & Yes & Yes \\
  \hline
 Alphabet Size& \multicolumn{3}{c|}{Binary} & Binary and polyphase \\
\hline
% Constellation size & \multicolumn{2}{c|}{$N$} &
%  $L$ & lcm$(N,n_P)$ \\ \hline
%%  1 & 2 & 3 & 4 & 6 & 7 & 8 & 9 \\ \hline
 \end{tabular}}
\end{table}
\renewcommand{\arraystretch}{1}

\subsection{Examples of CSS-Based Polyphase ZCZ Sequence Sets}
\label{section:CSS_examples} Two ZCZ sequence construction
examples based on CSS are given in this subsection.
\begin{example} ($\gcd(N_r,N')\neq1$)~
Let $N=16$, $N'=N_r=K=L=4$, $A'=(+++-\hspace{.1em})$, and
$B=(1000000100100100)$ and choose a collection of mutually
orthogonal complementary sets
$\mathcal{E}=\{\mathbf{E}_0,\mathbf{E}_1,\mathbf{E}_2,\mathbf{E}_3\}$
from \cite{complementary}, where
\begin{IEEEeqnarray}{rCl}
\mathbf{E}_0&=&\{ (+++\:+), (--+\:+), (-+-\:+), (+--\:+) \},\nonumber\\
\mathbf{E}_1&=&\{ (++-\:-), (---\:-), (-++\:-), (+-+\:-) \},\nonumber\\
\mathbf{E}_2&=&\{ (-+-\:+), (+--\:+), (+++\:+), (--+\:+) \},\nonumber\\
\mathbf{E}_3&=&\{ (-++\:-), (+-+\:-), (++-\:-), (---\:-) \}.~~~~~%\nonumber\\
\label{eq:L4_MOCSS}
\end{IEEEeqnarray}
Following the procedure of \textit{Corollary \ref{poly_by_comp}},
we obtain the bound-achieving binary $(64,4,8)$ ZCZ sequence set:
\begin{IEEEeqnarray*}{rCl}
C_0=(++++-++--+-+--++----+-
-+-+-+--+&+&\nonumber\\+++++--+-+-+
++--+++++--++-+---+&+&\hspace{-.1em}),\nonumber\\
C_1=(+   +   -   -   -   +   -   +   -   +   +   -   -   -   -   -
-   -   +   +   +   -        +   -   -   +   +   - - -   -
&-&\nonumber\\+ +   -   -   +   -   +   -   -   +   +   -   +   +
+ + + + - -   +   -   +   -   +   -   -   +
-   -   -   &-&\hspace{-.1em}),\nonumber\\
C_2=(-   +   -   +   +   +   -   -   +   +   +   +   +   -   - +
+   -   +   -   -   -        +   +   +   +   +   + + -   -
&+&\nonumber\\  - +   -   +   -   -   +   +   +   +   +   + - + +
- - + -   + -   -   +   +   -   -   -   -
+ -   - &+&\hspace{-.1em}),\nonumber\\
C_3=(-   +   +   -   +   +   +   +   +   +   -   -   +   -   + -
+   -   -   +   -   -      -   -   +   +   -   - + - +
&-&\nonumber\\ - +   +   -   -   -   -   -   +   +   -   - - + - +
- + + -   -   -   -   -   -   -   +   +   +   -   +
&-&\hspace{-.1em}).\nonumber
\end{IEEEeqnarray*}
With the same $B$, $A'$, and $N_r$ as those used in {\it Example
\ref{ex:BiDirect2}}, this set extends the ZCZ width without
changing the set cardinality.
\end{example}
\begin{example}
($\gcd(N_r,N')=1$)~Using the construction (\ref{eq:direct_prime})
with (\ref{eq:L4_MOCSS}), $A'=(W_3^0 W_3^2 W_3^0)$, $N'=3$,
$N_r=4$, and $L=4$, we can obtain a ZCZ sequence set $\mathbf{C}$
of the same (or larger) $T$ with a shorter sequence period $LN$
and slightly larger constellation:
\begin{IEEEeqnarray*}{rCl}
C_0=(&&W_6^0   W_6^0   W_6^0   W_6^0   W_6^0   W_6^3   W_6^3 W_6^0
W_6^5   W_6^2   W_6^5   W_6^2 W_6^3   W_6^3   W_6^0   W_6^0 W_6^0   W_6^0   W_6^0   W_6^0
W_6^2   W_6^5   W_6^5   W_6^2\nonumber\\
&& W_6^3   W_6^0   W_6^3   W_6^0 W_6^3   W_6^3   W_6^0   W_6^0
W_6^2
W_6^2   W_6^2   W_6^2 W_6^0   W_6^3   W_6^3   W_6^0 W_6^3   W_6^0 W_6^3   W_6^0 W_6^5
W_6^5   W_6^2   W_6^2),\nonumber\\
C_1=(&&W_6^0   W_6^0   W_6^3   W_6^3   W_6^0   W_6^3   W_6^0 W_6^3
W_6^5   W_6^2   W_6^2   W_6^5 W_6^3   W_6^3   W_6^3   W_6^3 W_6^0   W_6^0   W_6^3   W_6^3
W_6^2   W_6^5   W_6^2   W_6^5\nonumber\\
&& W_6^3   W_6^0   W_6^0   W_6^3 W_6^3   W_6^3   W_6^3   W_6^3
W_6^2
W_6^2   W_6^5   W_6^5 W_6^0   W_6^3   W_6^0   W_6^3 W_6^3   W_6^0 W_6^0   W_6^3 W_6^5
W_6^5   W_6^5   W_6^5),\nonumber\\
C_2=(&&W_6^3   W_6^0   W_6^3   W_6^0   W_6^3   W_6^3   W_6^0 W_6^0
W_6^2   W_6^2   W_6^2   W_6^2 W_6^0   W_6^3   W_6^3   W_6^0 W_6^3   W_6^0   W_6^3   W_6^0
W_6^5   W_6^5   W_6^2   W_6^2\nonumber\\
&& W_6^0   W_6^0   W_6^0   W_6^0 W_6^0   W_6^3   W_6^3   W_6^0
W_6^5
W_6^2   W_6^5   W_6^2 W_6^3   W_6^3   W_6^0   W_6^0 W_6^0   W_6^0 W_6^0   W_6^0 W_6^2
W_6^5   W_6^5   W_6^2),\nonumber\\
C_3=(&&W_6^3   W_6^0   W_6^0   W_6^3   W_6^3   W_6^3   W_6^3 W_6^3
W_6^2   W_6^2   W_6^5   W_6^5 W_6^0   W_6^3   W_6^0   W_6^3 W_6^3   W_6^0   W_6^0   W_6^3
W_6^5   W_6^5   W_6^5   W_6^5\nonumber\\
&& W_6^0   W_6^0   W_6^3   W_6^3 W_6^0   W_6^3   W_6^0   W_6^3
W_6^5
W_6^2   W_6^2   W_6^5  W_6^3   W_6^3   W_6^3   W_6^3 W_6^0   W_6^0 W_6^3   W_6^3 W_6^2
W_6^5   W_6^2   W_6^5).\nonumber
\end{IEEEeqnarray*}
\end{example}
It is worth mentioning that the above set cannot be obtained by
using the methods of \cite{comp_ex_1} and \cite{comp_ex_4} and,
moreover, although \textit{Corollary \ref{poly_by_comp}} promises
an $(LN,N_r,L(N'-1))=(48,4,8)$ family, $\mathbf{C}$ is actually a
$(48,4,9)$ one. The larger ZCZ is due to the inherit correlation
properties of MOCSS (\ref{eq:L4_MOCSS})
\begin{equation}
  \sum_{k=0}^{N_r-1} \psi_{E_{i,k}E_{j,|k\pm1|_{N_r}}}(\tau)=0
\end{equation}
for $\tau=\pm(L-1)$, $0\leq i<N_r$, and $0\leq j<N_r$.

%\vspace{1em}
\section{Conclusion}\label{section:conclusion}
Three new systematic approaches--a transform domain method and two
direct (time domain) synthesis methods--for generating ZCZ
sequence families have been presented in this paper. The transform
domain approach exploits the cross-correlation function's
transform domain representation and the recursive Kronecker
structure of a class of Hadamard matrices. The two other
approaches begin with simple binary basic ZCZ sequences. Through
progressively fine-tuning steps that include novel basic
sequence-based nonuniform upsampling of unitary matrices or a
collections of MOCSS's, we are able to obtain polyphase sequences
that meet various ZCZ requirements.

The basic sequences are used to ensure that the required ZCZ width
is satisfied during the upsampling process while the transform
domain approach uses the subperiodicity of the Hadamard product of
two transform domain sequences. The orthogonality among rows of
unitary matrices or MOCSS guarantees that the CC value of any two
member sequences at zero lag is zero as well. We take advantage of
the correlation-invariant property of the
filtering-by-perfect-sequence operation to convert a nonconstant
modulus sequence into a polyphase sequence. Judicious choices of
the basic and perfect sequences used and the associated upsampling
rate are crucial in this operation.

Our approaches are conceptually simple and require no
sophisticated algebra but, in some cases, offer more flexibilities
in either the choices of the sequence length, the ZCZ width and/or
the alphabet size needed. We are therefore able to produce
sequence families with the same parameters as those by earlier
proposals as well as some that are not achievable by related known
methods. Finally, for each approach, numerical examples have been
provided to further validate the proposed construction methods.

%\begin{appendix}
\appendices
\renewcommand{\thesection}{Appendix \Alph{section}}
\renewcommand{\theequation}{\Alph{section}.\arabic{equation}}
%\section{Proof of Lemma 4}
%\label{app:pf_partition_fundamental}
%%\begin{proof}
%\setstretch{1.1}
%\setcounter{equation}{0}
\section{Proofs of Theorem \ref{poly_by_direct} and Corollary
\ref{poly_by_comp}} \label{app:pf_direct} \setcounter{equation}{0}
Let $P=(p_0, p_1,\cdots, p_{N-1})$ be a row of $\mathbf{P}$ and
$C=P\circ A=(c_0,c_1,\cdots, c_{N-1})$, where
$c_n=\sum_{j=0}^{N-1} p_j a_{|j-n|_N}^*$. If we can show that, for
any $n \in [0,N-1]$, one and only one of the $N$ products $\{p_j
a_{|j-n|_N}^*: j=0,\cdots,N-1\}$ is nonzero, then, as both $A$ and
$P$ consist $0$'s and polyphase elements, $C$ is a polyphase
sequence as well. Because of the circular convolution nature of
the filtering operation ({\it Definition \ref{def:mod}}) and the
periodic run property of $P$, we have only to check if this single
nonzero product assertion is valid for $0 < n < N_r$.

For the first construction (\ref{eq:direct_prime}),
gcd$(N_r,N')=1$ and both $N_r$ and $N'$ are positive, hence
$\exists$ unique $a, b\in\mathbb{Z}$ such that $aN'+bN_r=1$, where
one of the integer coefficients $a$ or $b$ must be negative \cite{Bezout}.
Without loss of generality, we assume $b<0$ and multiply both
sides of the above B\'{e}zout's identity by $s$, $0< s <N_r$, to
obtain $saN'=s+sb'N_r$, $b'=-b > 0$. If $sb' \leq N'-1$ then $saN'
<N'N_r=N$ and $sa < N_r$; otherwise, subtract both sides by
$n_0N$, where $n_0=\left\lfloor\frac{sb'N_r}{N}\right\rfloor$ to
obtain $(sa-n_0N_r)N'=s+(sb'-n_0N')N_r$. For both cases, we have,
for each positive $s <N_r$, $\exists$ unique pair of positive
integers $(m,n)$, $0 < m \leq N_r-1$, $0 \leq n \leq N'-1$ such
that $mN'=s+nN_r$ mod $N$. That this property holds for $s=0$ is
obvious.

As for the second construction (\ref{eq:direct_notprime}), we
notice that the basic sequence admits the orthogonal tone
decomposition, $B=\sum_{\ell=0}^{d-1}B_\ell$, where
\begin{IEEEeqnarray}{rCl}
 B_\ell(i)=\left\{%
 \begin{array}{ll}
    b_i, & \hbox{$\ell L_0\leq i<(\ell+1)L_0$;} \\
    0, & \hbox{otherwise.} \\
 \end{array}%
 \right.
\end{IEEEeqnarray}
When $d=\text{gcd}(N_r,N')$, there exists positive integers $a$,
$b'$ such that $aN'=d+b'N_r$. Multiplying both sides by $s$,
$0\leq s<\frac{N}{d}$, we obtain
$(sa-n_0\frac{N_r}{d})N'=sd+(sb'-n_0\frac{N'}{d})N_r$, where
$n_0=\left\lfloor\frac{saN'}{L_0}\right\rfloor$. For all
$s\in\left\{0,1,\cdots, \frac{\lfloor N_r/N'\rfloor
N'+(N-N_r)}{d}\right\}$, $\exists$ a unique integer pair $(m,n)$,
$0\leq m<\frac{N_r}{d}$, $0\leq n<N'$ such that $mN'=%|sd+nN_r|_N=
sd+nN_r$ mod $N$, i.e., the sequence $B_0\circ A$ is identically
zero except at indices that are multiples of $d$ and the nonzero
terms are the products of two polyphase signals whence are
themselves polyphase signals.

Similarly, we can show that, for $\ell=1,2,\cdots,d-1$, the
sequence $B_\ell\circ A$, has nonzero polyphase terms at $nd-\ell$
only, where $n\in\mathbb{Z}$. Hence the sequence $B\circ A=
\sum_{\ell=0}^{d-1}B_\ell\circ A$ is a polyphase sequence.

To prove \textit{Corollary \ref{poly_by_comp}}, we first note that
the sequences generated differs from those generated by
\textit{Theorem \ref{poly_by_direct}} in that the perfect sequence
used in \textit{Corollary \ref{poly_by_comp}} is the $L$-fold
upsampled version of that used in \textit{Theorem
\ref{poly_by_direct}} while the unfiltered ZCZ sequences for the
former is an $L$-expanded version of those for the latter,
replacing each zero entry of $P$ by a length-$L$ string of zeros
and each nonzero entry by a complementary sequence $E_{ij}$ of
length $L$.

For the first construction of $B$ (\ref{eq:direct_prime}), we
immediately have, for $0 \leq s < N_r$, $\exists$ unique pair of
positive integers $(m,n)$, $0 < m \leq N_r-1$, $0 \leq n \leq
N'-1$ such that $mLN'=sL+nLN_r=k+nLN_r$. That is, in computing the
filtered sequence $C=G\circ A=\{c_k\}$, where $G=\{g_k\}\defeq
G_i$ and $c_k=\sum_{j=0}^{LN-1} g_j a_{|j-k|_{LN}}^*$, there is
only one nonzero term in the summands that add up to $c_k$, for
$k=sL, s=0,1,\cdots, N_r-1$. That this single nonzero convolution
term property holds for $sL < k < (s+1)L$ is obvious because of
the special structure of $G_i$. The proof for the case when the
second construction (\ref{eq:direct_notprime}) is employed follows
a similar line of argument.

\end{document}